\documentclass[10pt,twocolumn]{IEEEtran}

\usepackage{amsmath}
\usepackage{amsthm,amssymb,amsmath,bm}
\usepackage{subfigure}
\usepackage{gensymb}
\usepackage{amsfonts}
\usepackage{epsfig}
\usepackage{amssymb}
\usepackage{amsmath}
\usepackage{cite}
\usepackage[utf8x]{inputenc}
\usepackage{tabularx}
\usepackage[normalem]{ulem} 
\usepackage{soul}
\usepackage{color,soul}
\usepackage{subfigure}
\usepackage{multirow}
\usepackage{rotating}
\usepackage{graphicx}
\usepackage{tabularx}
\usepackage{array}
\usepackage{color,soul}
\usepackage{bm}
\usepackage{graphicx,dblfloatfix}
\usepackage{blindtext}

\newtheorem{remark}{Remark}
\usepackage{subfigure}
\usepackage{moreverb}
\usepackage{epsfig}
\usepackage{amsmath,amssymb,amsthm,mathrsfs,amsfonts,dsfont}
\usepackage{adjustbox,lipsum}
\usepackage{amsfonts}
\usepackage{epsfig}
\usepackage{amssymb}
\usepackage{amsmath}
\usepackage{amsthm}
\usepackage{subfigure}
\usepackage{multirow}
\usepackage{rotating}
\usepackage{graphicx}
\usepackage{tabularx}
\usepackage{array}
\usepackage{anyfontsize}
\usepackage{color,soul}
\usepackage{graphicx,dblfloatfix}
\usepackage{epstopdf}
\usepackage{blindtext}
\usepackage{amsmath}
\usepackage{amsthm,amssymb,amsmath,bm}
\usepackage{subfigure}
\usepackage{amsfonts}
\usepackage{epsfig}
\usepackage{amssymb}
\usepackage{amsfonts}
\usepackage{amsmath}
\usepackage{cite}
\usepackage{graphicx}
\usepackage{fancyhdr}
\usepackage{subfigure}
\usepackage[subfigure]{tocloft}
\usepackage[font={small}]{caption}
\usepackage{subfigure}
\usepackage{tabularx}
\usepackage{tcolorbox}
\usepackage{cite}

\usepackage{multirow}
\usepackage{rotating}
\usepackage{tabularx}
\usepackage{array}
\usepackage{dblfloatfix}
\usepackage{blindtext}

\usepackage[linesnumbered,ruled,vlined]{algorithm2e}
\SetKwInput{KwInput}{Input}
\SetKwInput{KwOutput}{Output}

\usepackage{amsthm,amssymb,amsmath,bm}
\usepackage{graphicx}
\usepackage{fancyhdr}
\usepackage{subfigure}
\usepackage[subfigure]{tocloft}
\usepackage[font={small}]{caption}
\usepackage{subfigure}
\usepackage{tabularx}
\usepackage{cite}
\allowdisplaybreaks
\usepackage[colorlinks,bookmarksopen,bookmarksnumbered,citecolor=blue,urlcolor=red]{hyperref}

\newtheorem{pro}{Proposition}

\usepackage{amsmath, amsthm, amssymb}
\usepackage{bbding}
\usepackage{booktabs}
\usepackage{graphicx}
\usepackage{epstopdf}
\usepackage{epsfig}
\usepackage{amssymb}
\usepackage{amsmath}
\usepackage{cite}
\usepackage{color,soul}
\usepackage{subfigure}
\usepackage{multirow}
\usepackage{rotating}
\usepackage{tabularx}
\usepackage{array}
\usepackage{dblfloatfix}
\usepackage{blindtext}
\usepackage[normalem]{ulem} 
\usepackage{soul} 

\begin{document}
		\title{ \huge Joint Radio Resource Allocation and 3D Beam-forming in MISO-NOMA-based Network: Profit Maximization for Mobile Virtual Network Operators}
			\author{\IEEEauthorblockN{Abulfazl Zakeri and  Nader Mokari, \IEEEmembership{Senior Member, IEEE}, and Halim Yanikomeroglu, \IEEEmembership{Fellow, IEEE}}\\	
		%
		\thanks{A. Zakeri and N. Mokari are with the Department of ECE,
			Tarbiat Modares University, Tehran, Iran (email: Abolfazl.zakeri@modares.ac.ir and nader.mokari@modares.ac.ir).
			H. Yanikomeroglu is with the Department of Systems and Computer
			Engineering, Carleton University, Ottawa, ON, K1S 5B6, Canada.
		}}	
	\maketitle
	\vspace{-2cm}
		%
	\begin{abstract}
		Massive connections and high data rate services are key players in 5G ecosystem and beyond. To satisfy the requirements of these types of services, 
		non orthogonal multiple Access (NOMA) 
		and 3-dimensional beam-forming (3DBF)  can be exploited.  
		In this paper, we devise a novel radio resource allocation and 3D multiple input single output (MISO) BF algorithm in NOMA-based heterogeneous networks (HetNets) at which our main aim is to maximize the profit of mobile virtual network operators (MVNOs).
		 \textcolor{black}{To this end, we consider multiple infrastructure providers (InPs) and MVNOs serving multiple users. 
		Each  InP has multiple access points as base stations (BSs) with specified spectrum and multi-beam antenna array in each transmitter  that share its spectrum with MVNO's users by employing NOMA.} To realize this, we formulate a novel optimization problem at which the main aim is to maximize the revenue of MVNOs, subject to resource limitations and quality of service (QoS) constraints. Since our proposed optimization problem is non-convex and mathematically intractable,
		we transform it into a convex one by introducing a new optimization variable and
		 converting the variables with adopting successive convex approximation.
		 More importantly,  the proposed solution is assessed and compared  with the alternative search method and the optimal solution \textcolor{black}{that is obtained with adopting the exhaustive search method.} 
		 In addition, it is studied from the computational complexity, convergence, \textcolor{black}{and performance perspective.}
		   Our simulation results demonstrate that NOMA-3DBF has better performance and increases system throughput and MVNO's revenue compared to orthogonal multiple access with 2DBF by approximately $64$\%. Especially, by exploiting 3DBF the MVNO's revenue is improved nearly $27$\% in contrast to 2DBF in  high order of antennas.
	\end{abstract}
	\maketitle
	\section{{introduction}} 
\subsection{State of The Art}
	 \IEEEPARstart{W}{ith} considering exponentially growth of wireless data traffic and diverse novel services such as enhanced mobile broadband, massive internet of thing (mIoT), and critical communication in the fifth  generation (5G) of wireless networks, efficient, flexible, and  on-demand resource utilization are become very important \cite{8642812, ding2017survey, 7414384}. 
	To fulfill these services requirements, high frequency band and advanced access technologies are proposed for new radio\footnote{New radio (NR) is the name of 5G random access technology.} of 5G access technology \cite{al2019sum, huaweisp,MedaiTek,el2018key}.
	In 5G, due to new emerging services and ultra high density of users, spectral and energy efficiency need be increased significantly. 
	To achieve these goals, advanced physical layer technologies, e.g., high order multi input and multi output (MIMO) and non-orthogonal multiple access (NOMA) are proposed and investigated for 5G \cite{el2018key, nadeem2018elevation, NOMA, 8695086, 8375979,vaezi2019interplay,8786250}. 
	NOMA-based networks have an array of benefits  such as  high spectral efficiency (SE) and accommodate massive connectivity compared to conventional ones multiple access techniques. \textcolor{black}{	Moreover, NOMA  reduces radio frequency chains and hardware cost \cite{power2018joint, NOMA}.}
	On the other hand, multiple antennas systems such as multiple input single output (MISO) can significantly improve SE and
	reliability. \textcolor{black}{In MISO systems,  these gains are achieved without increasing 
the  size, cost, and battery life, i.e., the lower level of processing requires less battery consumption at the receivers sides, overally.} 
	Exploiting joint NOMA and multiple antennas systems 
	  have  significant advantages in terms of improving 
	SE by utilizing each resource block more than one in each base station (BS) and applying spatial diversity, respectively \cite{lTEMIMOenhanced, 8648507}. 
 Furthermore, the beam pattern in physical layer has significant impacts on the performance of wireless network; especially 
  for the high losses of propagation at high frequency bands \cite{MedaiTek}.   	
	Three-dimensional beam-forming (3DBF) combining  the horizontal and vertical pattern allocation   
	 with large number of antennas is one of the most promising technologies for 5G \cite{kammoun2018design, shafin2014channel, cui2018optimal, 8718197, MedaiTek}.
	 As studied in \cite{7582424}, interference management is a key limiting factor in increasing the capacity of heterogeneous networks (HetNets). To tackle this issue,
	 by exploiting 3DBF, the vertical and horizontal beams are directed such that 
	 the power of intercell interference can be reduced significantly \cite{li2013dynamic}. Moreover, in contrast to two-DBF (2DBF), i.e., just  horizontal (azimuth) antenna pattern  is adjusted, 3DBF provides tracking of users in both the horizontal and vertical and has negligible side lobes \cite{8539578}. 
	 
Nevertheless, combining  multiple antennas systems  and NOMA with 3DBF can increase SE and capacity of wireless networks multiple times \cite{vaezi2019interplay} and address heterogeneous service requirements such as, low latency, high data rate, and reliability. Hence, these physical layer technologies are appropriate candidates to address the  requirements of 5G and beyond. Accordingly, studying 3DBF in NOMA-based MISO for HetNet from the performance, computational complexity, signaling overhead, and pricing perspectives is the main focus of this paper.
	\subsection{{Related Works}}
 \textcolor{black}{Related works on this article can be discussed in the two main categories:  1) NOMA-based MISO networks with various objective functions such as  throughput maximization and power consumption minimization; 2) 2DBF  and 3DBF in NOMA and (or) MISO-based networks.} 
 \subsubsection{{NOMA-based MISO Networks}}
	Due to the pivotal role of NOMA with MISO on the performance of physical layer of wireless networks, recently, many researches are done in this area 
	 \cite{xiao2018opportunistic,twodimention,MIMO-NOMA,mimo-nomacsi, al2019energy}. 
	Optimal  power allocation  for NOMA-based MISO satellite networks with power minimization is proposed in \cite{alhusseini2019optimal}. 
 Robust radio resource allocation in power-domain NOMA (PD-NOMA)-based networks with adopting matching theory is studied in  \cite{rezaei2019robust}. \cite{8686217} studies a radio resource allocation in MISO-enabled cloud radio access network. Moreover, energy efficiency for NOMA-based MISO in a single cell network has been studied in \cite{al2019energy1}. The authors in \cite{8695086}, propose a new power minimization problem in MISO-enabled PD-NOMA-based networks. They consider the minimum rate constraint with steering beam forming and user clustering variables. 
	 \subsubsection{{2DBF and 3DBF in MISO Networks}}
	Up to now, 2DBF and 3DBF multiple antennas systems \textcolor{black}{ are received much attention and are studied} in many researches from different aspects \cite{alavi2017robust, li2015sum, fan2017exploiting, lTEMIMOenhanced, 8718092}. 
	The authors in \cite{8642812}  propose  BF and power allocation with designing two optimization problems  for the terrestrial network. 
	A robust power minimization BF for NOMA-based systems with considering imperfect channel state information (CSI) is studied in \cite{alavi2017robust}. Two dimensional precoding paradigm is applied for 3D massive antennas systems in \cite{twodimention}.  
\textcolor{black}{In \cite{8718197},  the user experience\footnote{\textcolor{black}{User experience data rate is defined as the data rate that, under loaded conditions, is available with 95\% probability\cite{ITU}.}} throughput in relay systems with  exploiting 3DBF is evaluated.}  
	The authors in \cite{al2019sum}, study BF in NOMA-based MISO networks with considering two performance metrics, namely,  sum data rate and fairness. 	Two-step BF with considering a multiuser NOMA-based MISO system is proposed in \cite{8695086}. The first step is a BF nulling
	to interference reduction at users  in  power domain superimposed coding cluster and the second step is BF steering for the desired cluster to minimize power utilization.
 
 As can be concluded non of aforementioned works studies joint 3DBF in NOMA-based MISO in multicell networks  from the cost and performance perspectives.  Moreover, almost all of them consider single cell network \cite{al2019energy1,power2018joint,8119791,alavi2017robust,8695086}. On the other hand, 3DBF is very appropriate choice to reduce interference and losses of propagation at high frequency bands (e.g., mmWave communication) that is desired for 5G and beyond. 
\subsection{{Motivations and Contributions}}
 To the best of our knowledge,
 there is no work on 
 pricing model for 3DBF MISO with the NOMA technique in the HetNet framework.
 \textcolor{black}{On the other hand,  employing both 3DBF in massive antenna systems and NOMA not only have a significant improvement on SE, coverage, etc, \textcolor{black}{but also they can reduce the overall cost of wireless networks, especially in high frequency band communication.}
 	These reasons are motivated us} to study new pricing models with the joint of 3DBF multi user MISO and NOMA in HetNet by considering infrastructure as a service.

	In this paper, we propose a novel joint radio resource allocation, user association, and 3DBF optimization problem
	with NOMA-based MISO in HetNet. 
	 To this end, 
	we consider the case of multi infrastructure providers (InPs) which provide infrastructure as a service for multi mobile virtual
	network operators (MVNOs) and design a novel pricing-based optimization problem under some constraints. \textcolor{black}{The main aim of the proposed optimization problem is to maximize the revenue of MVNOs' users subject to the transmit power, subcarrier allocation, successive interference cancellation (SIC), and QoS constraints.} 
	 
	 Our main contributions are summarized as follows:
	\begin{itemize}
		\item 
		We propose a new 3D beam and radio resource allocation in downlink  of NOMA-based MISO HetNet with proposing a novel pricing model \textcolor{black}{as a mixed integer non-linear programing problem; with aiming to maximize the revenue of MVNOs' users}. \textcolor{black}{We consider multiple InPs that have own physical network and hardware including the BSs that are equipped with multiple transmitters. The InPs server the MVNO's users based on the service level agreement with MVNOs which is function of cost and revenue.} \textcolor{black}{Since our considered framework comprises infrastructure as a service with enabling-virtualization of resources, it is appropriate and applicable to employ network slicing as a  key enabler of 5G \cite{7926923} by considering each MVNO's resources as a slice.  }
	\item \textcolor{black}{We propose} three methods, namely, jointly solving continues and integer variables (JS-CIV), alternative search method (ASM\footnote{The basic idea behind the ASM method which is a well-known suboptimal solution for non-convex and NP-hard problems, is that the main optimization problem is divided into some convex subproblems and solved each of them iteratively.}), and optimal solution  to solve the optimization problem  which is non-convex and mathematically intractable.
\textcolor{black}{Based on the simulation results, JS-CIV outperforms ASM by approximately $7$\% and its the  optimality gap is nearly $8$\%.}	
	 \item 
	 We investigate the performance of the proposed system under different geographical locations as deployment scenarios, such as indoor hotspot, rural, etc. Moreover, we study 2D and 3D channels from the signaling overhead perspective.  
	 \item 
	 Our simulation results depict that the cost of power and MVNO's revenue incorporate the influence of transmission technologies, e.g., 3DBF with order of antennas and multiple access technique. 
	\end{itemize}
\subsection{{Paper Organizations}}
This article is arranged as follows. Section \ref{systemmodelandproblemformulations} displays the system model and the problem formulations. Section \ref{solutions} presents solution of the problem.
 The simulation
results are presented in Section \ref{simulations}. Finally, the concluding remarks of this paper is drawn in Section \ref{conclusions}.\\
	 {Notations}: \textcolor{black}{ Bold upper and lower case letters denote
	matrices and vectors, respectively. Transpose and conjugate
	transpose are indicted by ${(\cdot)}^T$ and $(\cdot)^{\dagger}$, respectively. $\|\cdot\|$
	represents the Euclidean norm, $| \cdot |$ denotes the absolute value, and $\odot$ indicates Hadamard or element-wise product. $\Bbb{E}\{.\}$ is the expectation operator.
	$\mathcal{C}$ represents the field of complex numbers.
	 $x_i$ denotes the $i$-th entry of vector $\bold{x}$ and $\bold{0}_{a}$ denotes all zero vector with size $a$. Moreover, 
$\Re\{.\}$  and $\Im\{.\}$ are the real and imaginary parts of the associated argument, respectively.}
	\section{{System Model and Problem Formulation}}\label{systemmodelandproblemformulations}
	\subsection{{System Model}}
	We consider a scenario with multiple InPs and multiple MVNOs 
	where each MVNO severs its users 
that are  placed on different locations over the total coverage area of the network. 
\textcolor{black}{In other words, we consider a virtualized case where the physical resources provided by several InPs are
	divided into several virtual resources each of which can be used by one MVNO. }
	 We denote the set of InPs by $i\in\mathcal{I}=\{1,\cdots,I\}$, the set of MVNOs by $v\in\mathcal{V}=\{1,\cdots,V\}$, and the set of macro BSs (MBSs) and femto BSs (FBSs) of InP $i$ by $f_i\in\mathcal{F}_i=\{0,\cdots,F_i\}$, where $f_i=0$ is the MBS. Hence, the set of the BSs is $\mathcal{F}=\cup_{i\in\mathcal{I}}\mathcal{F}_i$ and $F=|\mathcal{F}|$.
	 A typical example of the considered 
	  system model is depicted in Fig.\ref{system model}. As seen, each InP has some BSs that supports users of different MVNOs. Moreover, we consider a central unit as a management system  that manages the network centralizely, and  solves the considered optimization problem and obtains the corresponding optimization variables. 
	 We assume that the BSs are equipped with multiple antennas, i.e., $M_\text{T}$ antennas and the receivers are single antenna.
	 We consider that the set of all downlink users
	   are randomly distributed in the network and this set is denoted by $\mathcal{K_{\text{Total}}}=\{1,\cdots,K\}$, which is the union of all MVNOs users, i.e., $\mathcal{K_{\text{Total}}}=\cup_{v\in\mathcal{V}}\mathcal{K}_v$. In this paper, we focus on the  downlink scenario of 3DBF in NOMA-based MISO HetNet.
	 
	 Furthermore, we consider that the total bandwidth of each InP $i$, which is non overlapping with the other InPs, i.e., $\text{BW}_i$ is divided into $N_i$ subcarriers with equal bandwidth\footnote{In this paper, we assume that all InPs have the dedicated frequency band and the existing bandwidth is not shared with each other.}.
	  Moreover, we assume that each BS $f_i$ of InP $i$ has $N_i$ subcarriers. We also suppose that perfect CSI is available at each BS and the channel gain from BS $f_i$ to user $k$ over subcarrier $n_i$ is denoted by $\textbf{h}_{f_i,n_i,k}=[h_{f_i,n_i,k}^{m}]\in \mathcal{C}^{M_\text{T}\times 1}$. 
	  The beam vector assigned by BS $f_i$ to user $k$ over subcarrier $n_i$ is denoted by $\textbf{w}_{f_i,n_i,k}=[w_{f_i,n_i,k}^{m}]\in \mathcal{C}^{M_\text{T}\times 1}$. 
	  The maximum allowable transmit power of each FBS is $P_{\max}^{\text{FBS}}$ and 
	  for each MBS is $P_{\max}^{\text{MBS}}$. The other parameters are summarized in Table \ref{table-00}. 
	 	\begin{table}[h]
	 	\renewcommand{\arraystretch}{1.11}
	 	\centering
	 	\caption{Notations and their definitions}
	 	\label{table-00}
	 	\begin{tabular}{| c| l| }	
	 		\hline
	 		\textbf{Notation}& \textbf{Definition}\\\hline
	 		$\mathcal{K}_v$&Set of MVNO $v$ users\\ \hline
	 		$\mathcal{K}_{\text{Total}}$ &Set of  all users\\ \hline
	 		$\mathcal{N}_i$ &Set of InPs $i$ subcarriers\\ \hline
	 		$\mathcal{V}$ &Set of MVNOs\\ \hline
	 		$\mathcal{I}$ &Set of InPs\\ \hline
	 		$\mathcal{F}_i$ &Set of BSs in InP $i$\\ \hline
	 		$M_T$ & Number of transmit antennas for each BS\\ \hline
	 		${\rho_{f_i,n_i,k}}$ &Subcarrier assignment and user association indicator \\&  for user $k$ on subcarrier $n_i$ in BS $f_i$\\\hline
	 		$\gamma_{f_i,n_i,k}$& Received SINR
	 		 of user $k$ on\\& subcarrier $n_i$ from BS $f_i$\\\hline
	 		$L$& Maximum number of users that can be \\&assigned to each subcarrier \\\hline
	 		${\sigma^2_{f_i,n_i,k}}$ & Noise power at user $k$ on subcarrier  $n_i$ from BS $f_i$\\\hline
	 		$s_{f_i,n_i,k}$ &Transmit signal at user $k$ on subcarrier \\&$n_i$ from BS $f_i$\\\hline
	 		$w_{f_i,n_i,k}^{m}$& Beam weight for user $k$ on subcarrier $n_i$ in BS $f_i$ \\&on antenna $m$
	 		\\\hline	
	 	\end{tabular}
	 \end{table}
	\begin{figure*}
		\centering

		\includegraphics[width=.72\textwidth]{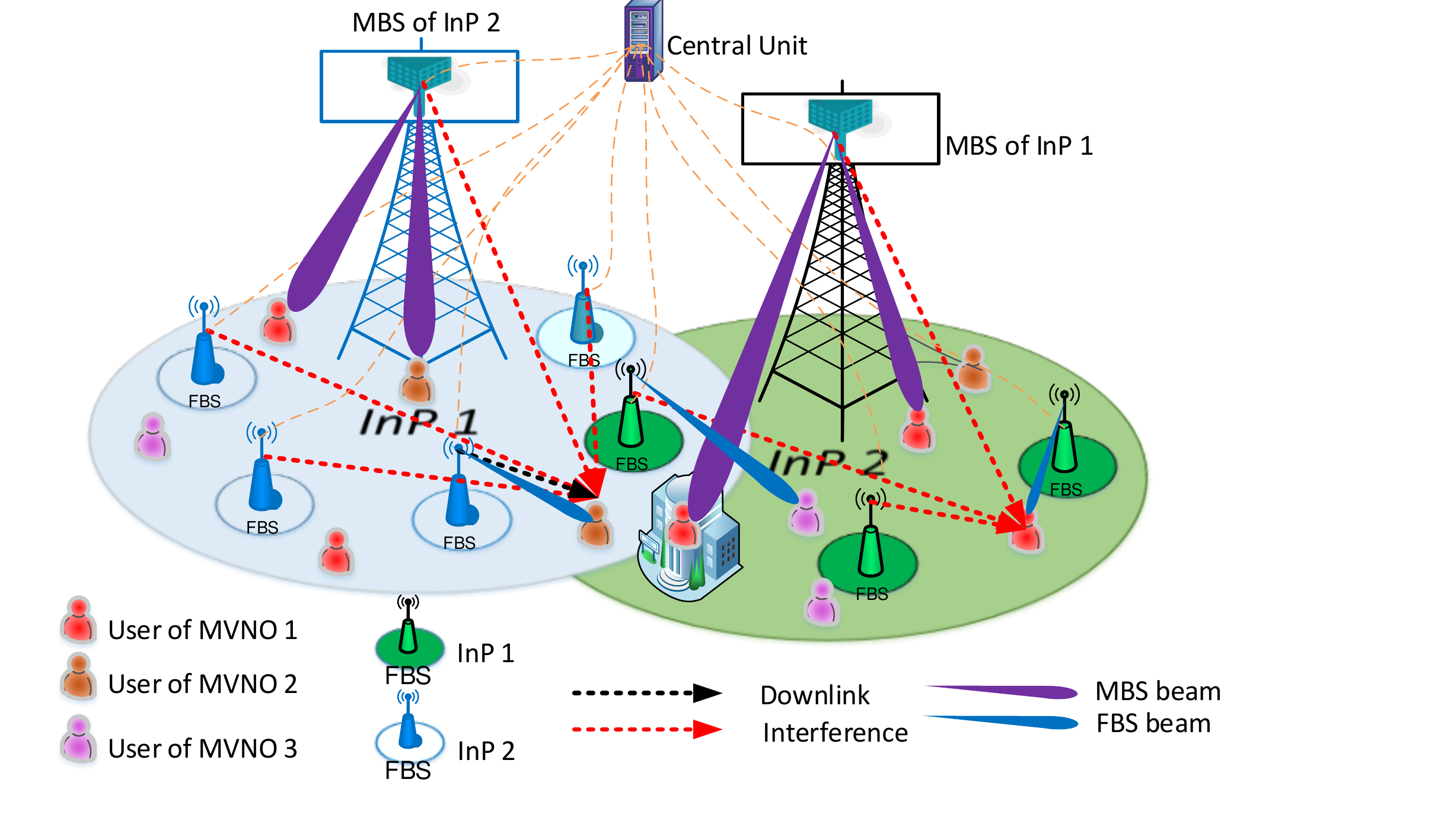}
		\caption{Considered system model. }
		\label{system model} 
	\end{figure*}
	\subsection{{2D and 3D Channel Models}}
	 In
	  3D directional channel model, the considered channel gain 
	  between BS $f_i$ and user $k$ on subcarrier $k$  is 
 formulated as follows \cite{li2015sum, rezaei2019robust},
	\textcolor{black}{
	\begin{align}
	\bold {h}_{f_i,n_i,k}={\beta}_{f_i,n_i,k}^{1/2}\boldsymbol{\Gamma}_{f_i,n_i,k},
	\end{align}
	where $ \boldsymbol{\Gamma}_{f_i,n_i,k}=[{\Gamma}_{f_i,n_i,k}^{m}]$ 
	is small scale fading 
	 and $	\beta_{f_i,n_i,k}$ denotes the large scale channel factor from  BS $f_i$ to user $k$ over subcarrier $n_i$ which is formulated in Section \ref{3dchannel}.} 
	\subsubsection{{3D Channel Model}}\label{3dchannel}
	To formulate the 3D directional large scale fading, i.e., ${\beta}_{f_i,n_i,k}$, first, we calculate 
	the horizontal
	and vertical angles between BS $f_i$ and user $k$, respectively, as follows:
	\begin{align}
	&\varphi_{f_i,k}(x_{f_i},y_{f_i})=\arctan\left(\frac{y_{f_i}-y_{k}}{x_{f_i}-x_{k}}\right),\\
	&\theta_{f_i,k}(x_{f_i},h_{f_i})=\arctan\left(\frac{ h_{f_i}}{x_{f_i}-x_{k}}\right),
	\end{align}
	where $(x_{f_i},y_{f_i})$ are 2D coordination of BS $f_i$, $h_{f_i}$ is the height of BS $f_i$,  and $(x_k,y_k)$ is 2D coordination of user $k\in \mathcal{K_{\text{Total}}}$.
	In order to obtain the large scale fading, the antenna gain\footnote{All gain values are in dB.} over subcarrier $n_i$ from BS $f_i$ to user $k$ is calculated by \eqref{3D gain}, \cite{lTEMIMOenhanced},
	\begin{align}
&	G_{f_i,n_i,k}(\varphi_{f_i,k},\theta_{f_i,k})=G_{f_i,n_i,k}^{\text{Hor}}(\varphi_{f_i,k})+G_{f_i,n_i,k}^{\text{Ver}}(\theta_{f_i,k}),\label{3D gain}
	\end{align}
	where $G_{f_i,n_i,k}^{\text{Hor}}(\varphi_{f_i,k})$ and $G_{f_i,n_i,k}^{\text{Ver}}(\theta_{f_i,k})$ are horizontal\footnote{The terms of horizontal and azimuth are used synonymously.} and vertical\footnote{The terms of vertical and elevation are used synonymously.} antenna patterns which are given by  
	 \cite{nadeem2018elevation}, \cite{fan2017exploiting}, \cite{kammoun2018design, nam2015full, li2015sum},
	\begin{align}
	\label{ghor}
	&G_{f_i,n_i,k}^{\text{Hor}}=G_{\max,\text{E}}-\min\left(12\left(\frac{\varphi_{f_i,k}-\hat{\varphi}_{0}}{\hat{\varphi}_{3dB}}\right)^2,\text{SLL}_{\text{Az}}\right),\\\label{gver}
	& G_{f_i,n_i,k}^{\text{Ver}}=-\min\left(12\left(\frac{\theta_{f_i,k}-\theta_{f_i,n_i,k}^{\text{tilt}}}{\hat{\theta}_{3dB}}\right)^2,\text{SLL}_{\text{El}}\right),
	\end{align}
	where $\hat{\varphi}_{3dB}$ and $\hat{\theta}_{3dB}$ depict the half-power beam-width in the azimuth and the elevation patterns, respectively. Whereas $G_{\max,\text{E}}$ is the maximum directional element
	gain at the antenna bore-sight, $\text{SLL}_{\text{Az}}$ and $\text{SLL}_{\text{El}}$ are the azimuth and elevation slide lobe levels, respectively. Furthermore, $\hat{\varphi}_{0}$ indicates the fixed orientation angle of each BS array boresight relative to the $y$-axis.  $\theta_{f_i,n_i,k}^{\text{tilt}}$ denotes the optimization variable tilt of BS $f_i$ and subcarrier $n_i$ for user $k$ which is measured between the direct line passing across the peak
	of the beam and horizontal plane.

	Motivated by the previous discussion, we can formulate
	the large-scale fading function $\beta_{f_i,n_i,k}$ which 
	is composed
	of path loss and 3D antenna gain and is given by \cite{kammoun2018design,li2015sum}
	\begin{align}\label{3dgain}
		\beta_{f_i,n_i,k}=d_{f_i,k}^{-\varrho}10^{\frac{G_{f_i,n_i,k}}{10}},
	\end{align}
	where
	 $d_{f_i,k}^{-\varrho}$ expresses the path loss, $d_{f_i,k}$ is the distance between BS $f_i$ and user $k$. Moreover, $\varrho$ is the path loss exponent and $G_{f_i,n_i,k}$ represents the gain of antenna and is given by \eqref{3D gain}.
	\subsection{{Radio Resource Allocation }}
	 In this section, we formulate the radio resource allocation problem based on the pricing model. To this end,
	we define the subcarrier assignment and BS selection variable, i.e., $\rho_{f_i,n_i,k}\in\{0,1\}$ with $\rho_{f_i,n_i,k}=1$, if user $k$ is scheduled to receive information from BS $f_i$ over subcarrier $n_i$, and otherwise  $\rho_{f_i,n_i,k}=0$.
	 We assume that the power of information symbol of user $k$ from BS $f_i$ over subcarrier $n_i$  is denoted by  ${s}_{f_i,n_i,k}$ is normalized to one, i.e., $\mathbb{E}(|s_{f_i,n_i,k}|^2)=1$. 
	  Thus, the signal transmitted by BS $f_i$ to user $k$ on subcarrier $n_i$ is given by
	\begin{align}\label{transmitedsignalbsbiuserkv}
	\textbf{x}_{f_i,n_i,k}= \rho_{f_i,n_i,k} \textbf{w}_{f_i,n_i,k}{s}_{f_i,n_i,k}.
	\end{align}	
	Note that each user should be assigned to at most one BS. This assumption can be taken into account by the following two constraints:
	\begin{align}\label{constraintoneuseronetransmitereachusereachsubcarrier1}\nonumber
	& \rho_{f_i,n_i,k}+\rho_{f_j,n_j,k}\le 1,\forall i\neq j, i,j \in \mathcal{I},\\&n_i \in \mathcal{N}_i, n_j \in \mathcal{N}_j,f_i,f_j\in\mathcal{F},\forall{k}\in\mathcal{K}_{\text{Total}},\\\label{constraintoneuseronetransmitereachusereachsubcarrier2}\nonumber& \rho_{f_i,n_i,k}+\rho_{f'_i,n'_i,k}\le 1, \forall i \in \mathcal{I},\,\\&\forall f_i \neq f'_i,f_i,f'_i\in \mathcal{F}_i,n_i,n'_i \in \mathcal{N}_i,\forall{k}\in\mathcal{K}_{\text{Total}}.
	\end{align}
	Constraint \eqref{constraintoneuseronetransmitereachusereachsubcarrier1} ensures that each user can be associated to at most one InP's network and constraint \eqref{constraintoneuseronetransmitereachusereachsubcarrier2} guarantees that each user can be assigned to at most one BS that is associated InP.
	For the subcarrier assignment, based on the NOMA approach in each cell, each subcarrier can be assigned to at most $L$ users.	That means, at most $L$ users could be scheduled for each subcarrier based on the following constraint:
	\begin{align}\label{constraintsubcarrierreuse1}
	&\sum_{k\in\mathcal{K}} \rho_{f_i,n_i,k}\leq L, \forall n_i\in\mathcal{N}_i,f_i\in\mathcal{F}_i,i\in\mathcal{I}.
	\end{align}
	Based on these definitions, the received signal of user $k$ which is assigned to BS $f_i$ for InP $i\in\mathcal{I}$ over subcarrier $n_i$ is given by \eqref{receivedsignalkv},
	\begin{figure*}[t]
	\begin{align}\label{receivedsignalkv}
&y_{f_i,n_i,k}=
	\underbrace{\textbf{h}^{\dagger}_{f_i,n_i,k} \textbf{x}_{f_i,n_i,k}}_\text{a}
	+\underbrace{\sum_{k'\in \mathcal{K}_{\text{Total}}\backslash k, k'\ge k}\textbf{h}^{\dagger}_{f_i,n_i,k'} \textbf{x}_{f_i,n_i,k'}}_\text{b}
	+\underbrace{\sum_{f'_i\in\mathcal{F}_i,f'_i\neq f_i }\sum_{k''\in\mathcal{K}_{\text{Total}}}\textbf{h}^{\dagger}_{f'_i,n_i,k} \textbf{x}_{f'_i,n_i,k''}}_\text{c}+\sigma^{2}_{f_i,n_i,k},
	\end{align}
	\hrule
\end{figure*}	
	where (a) is the desired signal and term (b) comes from the NOMA technique which is the interference from users with higher order in SIC ordering
	and (c) comes from inter-cell interference, respectively. In addition, for the sake of notational simplicity, we use  $k'\ge k$ to indicate that in SIC ordering, user $k'$ has higher order than that of user $k$, i.e., $\mid\textbf{h}^{\dagger}_{f_{i},n_i,k'}\textbf{w}_{f_{i},n_i,k'}\mid^2\geq
	\mid\textbf{h}^{\dagger}_{f_{i},n_i,k}\textbf{w}_{f_{i},n_i,k}\mid^2$, in the rest of this paper. Moreover, $\sigma^{2}_{f_i,n_i,k}$ is the power of
	additive white Gaussian noise (AWGN). 
	Using these definitions, the SINR of user $k$ is obtained as follows
	\begin{align}\label{sinrbikv}
	\gamma_{f_i,n_i,k}=\frac{\rho_{f_{i},n_i,k}\mid\textbf{h}^{\dagger}_{f_i,n_i,k}\textbf{w}_{f_i,n_i,k}\mid^2}{
		\text{Int}_{f_{i},n_i,k}^{\text{NOMA}}+\text{Int}_{f_{i},n_i,k}^{\text{Cell}}+\sigma^{2}_{f_i,n_i,k}},
		\end{align}
	where $\text{Int}_{f_{i},n_i,k}^{\text{NOMA}}$ and $\text{Int}_{f_{i},n_i,k}^{\text{Cell}}$ are \textcolor{black}{the intra-cell (NOMA) and the inter-cell interference and are given by} 
	\begin{align}&
	\text{Int}_{f_{i},n_i,k}^{\text{NOMA}}=\sum_{k'\in\mathcal{K}_{\text{Total}}\backslash k,\,k'\ge k} \rho_{f_{i},n_i,k'} 
	\mid\textbf{h}^{\dagger}_{f_{i},n_i,k}\textbf{w}_{f_{i},n_i,k'}\mid^2,	\label{nomainter}
	\\
	&\text{Int}_{f_{i},n_i,k}^{\text{Cell}}=
	\sum_{f'_{i}\in\mathcal{F}_{i}\backslash f_i}\sum_{k''\in\mathcal{K}_{\text{Total}}} \rho_{f'_{i},n_i,k''} \mid\textbf{h}^{\dagger}_{f'_{i},n_i,k}\textbf{w}_{f'_{i},n_i,k''}\mid^2,\label{intercell}
	\end{align}
\textcolor{black}{	Therefore, the achievable rate of user $k$ associated with BS $f_i$ on subcarrier $n_i$ is given by 
	\begin{align}\label{ratebikv}
\nonumber r_{f_i,n_i,k}=&\sum_{n_i\in\mathcal{N}_i}\log(1+\gamma_{f_i,n_i,k}),~ \forall f_i\in\mathcal{F}_i, i\in\mathcal{I}, \\&k\in\mathcal{K}_{\text{Total}}, n_i\in\mathcal{N}_i.
	\end{align}
	Hence,  the total rate of user $k$ associated with BS $f_i$ is $r_{f_i,k}=\sum_{n_i\in\mathcal{N}_i}r_{f_i,n_i,k}$.}
	\subsection{{Objective Function and Optimization Problem Formulation}}
 In order to include the cost in the objective function, we assume that each InP charges the MVNOs which are using its infrastructures based on the amount of the resources consumed by the users of these MVNOs. We also assume, the amount of charge of each MVNO $v\in\mathcal{V}$ is proportional to the amount of transmit power consumed by its users, i.e.,
	\begin{align}\label{chargeMVNOv}
	C^\text{Total}_v=&\sum_{i\in\mathcal{I}}\sum_{f_{i}\in\mathcal{F}_{i}}\sum_{n_i\in\mathcal{N}_i}\sum_{k\in\mathcal{K}_{v}} \rho_{f_{i},n_i,k}C_{i,f_{i}}\|\textbf{w}_{f_{i},n_i,k}\|^2,\,
	\end{align}
	where $C_{i,f_{i}}$ is the amount of charge per unit of power which should be paid when connecting to BS $f_i$ in InP $i$.	
	Furthermore, each MVNO charges its users based on the data rate provided to it, and hence, the income of MVNO is given by
	\begin{align}\label{incomeMVNOv}
	G^\text{Total}_v=\sum_{i\in\mathcal{I}}\sum_{f_{i}\in\mathcal{F}_{i}}\sum_{k\in\mathcal{K}_{v}}G_{v,k}r_{f_i,k},\,\forall v\in\mathcal{V},
	\end{align}
	where $G_{v,k}$ is the amount of charge per unit of data rate which should be paid to MVNO $v$ by user $k$. Hence, we formulate the revenue of MVNO $v$ as follows 
	\begin{align}
	\label{revenueMVNOv}
	U^\text{Total}_v=G^\text{Total}_v-C^\text{Total}_v.
	\end{align}
	Based on these assumptions and definitions, our proposed optimization problem is maximizing the total revenue of the MVNOs under 3DBF, BS, and subcarrier allocation variables and network constraints is stated as follows:
		\begin{subequations}\label{opt0}
			\begin{align}
		&	\max_{\bold {W}, \boldsymbol{\rho},\boldsymbol{\theta}} O(\bold {W},\boldsymbol{\rho},\boldsymbol{\theta})=\sum_{v\in\mathcal{V}}U^\text{Total}_v
			\\
		\nonumber	& \bold{s.t.}\,\,\\\label{opt2}&G^\text{Total}_v \geq C^\text{Total}_v,\forall v\in\mathcal{V},
			\\\nonumber\label{opt3}&
			\gamma_{f_i,n_i,k}\ge
			\gamma_{f_i,n_i,k'},
			\\&
			\forall k\ge k',\, \forall k\in\mathcal{K}_{\text{Total}},\forall i\in\mathcal{I},f_i\in\mathcal{F}_i,\forall n_i\in\mathcal{N}_i,
			\\\label{opt4}\nonumber
			& \sum_{k\in\mathcal{K}_{\text{Total}}}\sum_{n_i\in \mathcal{N}_i}\rho_{f_i,n_i,k}\|\bold {w}_{f_i,n_i,k}\|^2\le P_{\max}^{\text{FBS}},
			\\&\forall i\in\mathcal{I},\, f_i\in\mathcal{F}_i,\ f_i\neq 0,
			\\\label{opt5}\nonumber
			& \sum_{k\in\mathcal{K}_{\text{Total}}}\sum_{n_i\in\mathcal{N}_i}\rho_{f_i,n_i,k}\|\bold {w}_{f_i,n_i,k}\|^2\le P_{\max}^{\text{MBS}},
			\\
	&	\forall i\in \mathcal{I},\,f_i=0,
			\\\label{opt6}
			&\sum_{i\in\mathcal{I}}\sum_{f_i\in\mathcal{F}_i}\sum_{k\in\mathcal{K}_v}\sum_{n_i\in \mathcal{N}_i}r_{f_i,n_i,k}\ge R_{\min}^{v},\,\forall v\in\mathcal{V},
			\\&\label{opttild}\nonumber
		\theta_{f_i,n_i,k}^{\text{tilt},\min}\le	\theta_{f_i,n_i,k}^{\text{tilt}}\le\theta_{f_i,n_i,k}^{\text{tilt},\max},\,\\&\forall i\in\mathcal{I},f_i\in\mathcal{F}_i,\forall n_i\in\mathcal{N}_i,\forall k\in\mathcal{K}_{\text{Total}},		 
			\\&\label{opt7}\nonumber
			\rho_{f_i,n_i,k}+\rho_{f_j,n_j,k}\le 1,\forall i\neq j, i,j \in \mathcal{I},f_i\in\mathcal{F}_i,\\&f_j\in\mathcal{F}_j,n_i\in \mathcal{N}_i,n_j \in \mathcal{N}_j,\forall k\in\mathcal{K}_{\text{Total}},
			\\\label{opt8}\nonumber& 
			\rho_{f_i,n_i,k}+\rho_{f'_i,n'_i,k}\le 1,
			\forall f_i \neq f'_i, i \in \mathcal{I},\\&f_i,f'_i\in \mathcal{F}_i,n_i,n'_i \in \mathcal{N}_i,\forall k\in\mathcal{K}_{\text{Total}},
			\\ \label{opt9} &
			\sum_{k\in\mathcal{K}_{\text{Total}}} \rho_{f_i,n_i,k}\leq L, \forall n_i\in\mathcal{N}_i,f_i\in\mathcal{F}_i,i\in\mathcal{I},
			\\ \label {opt10}
			&
			\rho_{f_{i},n_i,k}\in \{0,1\},
			\forall n_i\in\mathcal{N}_i,f_i\in\mathcal{F}_i,i\in\mathcal{I}, k \in\mathcal{K_{\text{Total}}},
			\end{align}
		\end{subequations}
	where $\bold{W} = [\bold{w}_{f_i,n_i,k}]$, $\boldsymbol{\rho}=[\rho_{f_{i},n_i,k}]$, and $\boldsymbol{\theta}=[\theta_{f_i,n_i,k}^{\text{tilt}}]$. 
	 The constraints \eqref{opt2} and  \eqref{opt3} ensure that each MVNO has a benefit and SIC is performed correctly\footnote{This constraint guarantees that user $k$ is able to correctly decode the signal
	 	of other users on the same subcarrier (\underline{successful} interference
	 	cancellation) \cite{8-8115155}, \cite{7812683}.}, respectively. The constraints \eqref{opt4} and \eqref{opt5} show
	the available transmit power budget at each MBS and FBS, respectively. The constraint \eqref{opt6} guarantees the	 minimum data rate requirement of each MVNO from QoS perspective. The constraint \eqref{opttild} is used for the upper and lower bounds of 3DBF tilt.
	 The constraints  \eqref{opt7}-\eqref{opt10} are BS selection and subcarrier assignment limitations in which  
	\eqref{opt7} and \eqref{opt8}  indicate each user is only assigned to one BS and \eqref{opt9} demonstrates that each subcarrier can be allocated to at most $L$ users, simultaneously.
	\section{{Solution of Optimization Problem}}\label{solutions}
	The proposed optimization problem \eqref{opt0} is a non-linear programming problem incorporating both integer and continuous variables. Moreover, due to the non-concavity of the objective function and constraints \eqref{opt2}, \eqref{opt3}, and \eqref{opt6}, it
	is not-convex and intractable. Hence, the well-known convex optimization methods cannot be used, directly. Hence, we propose a new solution method, namely, JS-CIV, 
which transforms the original non-convex problem into a convex one. To this end, firstly, we introduce new variables, then merge the integer variable with continues variable,  and exploit the SCA technique. By these transformations, the well-known traditional sub-optimal solution, i.e., ASM, that is widely applied for  solving non-convex and NP-hard problems is not required to solve problem \eqref{opt0}. \textcolor{black}{The main two steps of JS-CIV are in the following.}
\textcolor{black}{	\subsection{{ Step-one (\underline{Converting and introducing  auxiliary variables})}}
	 We merge the integer variable (subcarrier and user association), i.e., $\rho_{f_{i},n_i,k}$ with the 3D beam weighted variable, i.e., $\bold {w}_{f_i,n_i,k}$ as follows:\\
	Clearly, from \eqref{opt7}, \eqref{opt8}, and \eqref{opt10}, we obtain that, if $\rho_{f_i,n_i,k}=1$, then $\rho_{f'_i,n'_i,k}=0$ and 
	$\rho_{f_j,n_j,k}=0, \forall i,j\in{\mathcal{I}},i\neq j,\,f_i\neq f'_i\in\mathcal{F}_i,n_i\in\mathcal{N}_i,k\in\mathcal{K}_{\text{Total}}$. 
	Therefore, if $\bold {w}_{f_i,n_i,k}\neq 	\bold {0}_{M_T}$ (\underline{All zero vector with size $M_T$}),  then $\bold {w}_{f'_i,n'_i,k}=	\bold {0}$ and 
	$\bold {w}_{f_j,n_j,k}$. 
	Due to the fact that, if subcarrier $n_i$ is not allocated to user $k$, then power and beam values on subcarrier $n_i$, i.e., $\bold{w}_{f_i,n_i,k}$ is zero.}
		 	\textcolor{black}{	\begin{pro}
			Assume that $\bold{w}_{1}$ and $\bold{w}_2$ are vectors with complex values and same sizes. Then, we have
			\begin{align}
			\bold{w}_{1}\odot \bold{w}_{2}=\bold{0} \Leftrightarrow \bold{w}_{1}=\bold{0} ~~\text{or}~~ \bold{w}_{2}=\bold{0}.
			\end{align} 
		\end{pro}
	\begin{proof}
		To prove this, 
		we consider $z_1$ and $z_2$ as the elements of  $\bold{w}_{1}$ and $\bold{w}_{2}$, respectively. Meanwhile, we easily can demonstrate  that $z_1.z_2=0\Rightarrow z_1=0~\text{or}~ z_2=0.$ One way to prove is to  consider that there exist $z_1\neq 0$ and $z_2\neq 0$ such that $z_1.z_2=0$. We infer that $z_1=\frac{0}{z_2}=0$, is contradiction. This can extend to all elements of $\bold{w}_{1}\odot \bold{w}_{2}=\bold{0}$.
	\end{proof}
	According to the above explanations,
	the user association constraints \eqref{opt7} and \eqref{opt8} can be replaced by the following continues constraints, equivalently:
	\begin{align}\label{bsass}\nonumber&
	\bold {w}_{f_i,n_i,k}\odot\bold {w}_{f'_i,n'_i,k}=\bold{0},\\& \forall i\in\mathcal{I}, f_i\neq f'_i\in\mathcal{F}_i,\,n_i,n'_i\in\mathcal{N}_i,k\in\mathcal{K}_{\text{Total}},
	\\\nonumber&
	\bold {w}_{f_i,n_i,k}\odot\bold {w}_{f_j,n_j,k}=\bold{0}, 
	\forall i\neq j, i,j \in \mathcal{I},f_i\in\mathcal{F}_i,\\&f_j\in\mathcal{F}_j,\forall n_i\in \mathcal{N}_i,n_j \in \mathcal{N}_j,\forall k\in\mathcal{K}_{\text{Total}}.\label{bsass1}
	\end{align}
	Similarly, based on constraints \eqref{opt9} and \eqref{opt10} and assumption of $L=2$,
	 if $\rho_{f_{i},n_i,k}=1$ and $\rho_{f_i,n_i,c}=1$, then, $\rho_{f_i,n_i,q}=0, \forall k,c,q\in\mathcal{K}_{\text{Total}},\,q\neq k, q \neq c $. Thus, if $\bold {w}_{f_i,n_i,k}\neq 	\bold {0}$ and $\bold {w}_{f_i,n_i,c}\neq	\bold {0}$, then, 
	$\bold {w}_{f_i,n_i,q}=	\bold {0}$.
	 Therefore, the subcarrier assignment constraints \eqref{opt9} and \eqref{opt10}, equivalently, can be replaced as follows:}
	\begin{align}\label{subass}
\nonumber&	\big(\bold {w}_{f_i,n_i,k}\odot\bold {w}_{f_i,n_i,c}\big)\odot\bold {w}_{f_i,n_i,q}=\bold{0}, 
	\\&
	\forall i\in\mathcal{I}, f_i\in\mathcal{F}_i,\,n_i\in\mathcal{N}_i,k, c, q\in\mathcal{K}_{\text{Total}}, k\neq c\neq q.
	\end{align}
\textcolor{black}{	\begin{remark}
	Note that $L=2$, is more practical and attracted a lot of attentions in both industry \cite{3gppp} and academia \cite{8786250,7273963}, \cite{7812683}, \cite{3-7362734}. However, we can generalize it to $L=3,4,\dots$ with accepting high complexity by following the same line of the proposed methods. For example, if $L=3$, we have $\Big(\big(\bold {w}_{f_i,n_i,k}\odot\bold {w}_{f_i,n_i,c}\big)\odot\bold {w}_{f_i,n_i,q}\Big)\odot {w}_{f_i,n_i,g}=\bold{0}.$
\end{remark}}
	By the proposed transformation, the integer variable is rewritten as a continues variable by considering its related new constraints.
\textcolor{black}{Moreover, to make  problem \eqref{opt0} more tractable, we introduce a new slack variable as $\boldsymbol{\lambda} = [\lambda_{f_i,n_i,k}]$. By this, the rate function \eqref{ratebikv} can be rewritten as $r_{f_i,n_i,k}=\log(\lambda_{f_i,n_i,k})$. To keep SINR is positive, we should consider that $\lambda_{f_i,n_i,k}\ge 0$.	
	Therefore, the original problem \eqref{opt0} can be \textcolor{black}{ reformulated} as follows \cite{7070667}:
		\begin{subequations}\label{opt20}
			\begin{align}\label{opt21}
&		\max_{\bold {W},\boldsymbol{\theta},\boldsymbol{\lambda}}
	O(\bold {W},\boldsymbol{\theta},\boldsymbol{\lambda})=\sum_{v\in\mathcal{V}} G^\text{Total}_v-C^\text{Total}_v,
		\\\nonumber
&	\bold{s.t.}\\\label{opt22}	&
		G^\text{Total}_v \geq C^\text{Total}_v,\forall v\in\mathcal{V},
		\\&\label{opt23}
		\lambda_{f_i,n_i,k}\ge 0,
		\forall i\in\mathcal{I},\forall f_i\in\mathcal{F}_i,\,\forall n_i\in\mathcal{N}_i, \forall k\in\mathcal{K}_{\text{Total}},
		\\&
		\label{opt24}
			1+\gamma_{f_i,n_i,k}\ge \lambda_{f_i,n_i,k},
		 \forall f_i\in\mathcal{F}_i, \forall n_i\in\mathcal{N}_i, \forall k\in\mathcal{K}_{\text{Total}},
		\\&\label{opt25}\nonumber
	\lambda_{f_i,n_i,k}\ge \lambda_{f_i,n_i,k'},
	\\&
	\forall k\ge k',\, \forall k\in\mathcal{K}_{\text{Total}},\forall i\in\mathcal{I},f_i\in\mathcal{F}_i,\forall n_i\in\mathcal{N}_i,
	\\\label{opt26}
	&
	 \sum_{k\in\mathcal{K}_{\text{Total}}}\sum_{n_i\in \mathcal{N}_i}\|\bold {w}_{f_i,n_i,k}\|^2\le P_{\max}^{\text{FBS}},
	\forall i\in\mathcal{I},\, f_i\in\mathcal{F}_i,\ f_i\neq 0,
	\\\label{opt27}
	& \sum_{k\in\mathcal{K}_{\text{Total}}}\sum_{n_i\in\mathcal{N}_i}\|\bold {w}_{f_i,n_i,k}\|^2\le P_{\max}^{\text{MBS}},
	\forall i\in \mathcal{I},\,f_i=0,
	\\\label{opt28}
	&\sum_{i\in\mathcal{I}}\sum_{f_i\in\mathcal{F}_i}\sum_{k\in\mathcal{K}_v}\sum_{n_i\in \mathcal{N}_i}\log(\lambda_{f_{i},n_i,k})\ge R_{\min}^{v},\, \forall v\in\mathcal{V},
	\\&\nonumber 
	\eqref{bsass},\,\eqref{bsass1},\eqref{subass},\eqref{opttild},
		\end{align}
	\end{subequations}
	where $\bold{W} = [\bold{w}_{f_i,n_i,k}]$, $\boldsymbol{\lambda}=[\lambda_{f_{i},n_i,k}]$, and $\boldsymbol{\theta}=[\theta_{f_i,n_i,k}^{\text{tilt}}]$.
The optimization problem \eqref{opt20} is still non-convex, due to constraints \eqref{opt24} and \eqref{bsass}-\eqref{subass} are non-convex. We convert it into convex one as explained next.
\subsection{Step-two (\underline{\textit{Approximation and changing variables}})}
By substituting \eqref{sinrbikv} into constraint \eqref{opt24} and some mathematical
manipulations, we can rewrite \eqref{opt24}  by \eqref{sinr}.
	\begin{figure*}[t]
	\begin{align}\label{sinr}
	&\frac{\mid\textbf{h}^{\dagger}_{f_{i},n_i,k}\textbf{w}_{f_{i},n_i,k}\mid^2}{\big(\lambda_{f_i,n_i,k}-1\big)}\ge
	\sum_{k'\in\mathcal{K}_{\text{Total}}\backslash k,k'\succ k} \mid\textbf{h}^{\dagger}_{f_{i},n_i,k}\textbf{w}_{f_{i},n_i,k'}\mid^2 
	+\sum_{f'_{i}\in\mathcal{F}_{i}\backslash f_i}\sum_{k'\in\mathcal{K}_{\text{Total}}\backslash k'} \mid\textbf{h}^{\dagger}_{f'_{i},n_i,k}\textbf{w}_{f'_{i},n_i,k'}\mid^2+{\sigma_{f_i,n_i,k}^{2}},
	\end{align}\hrule
	\end{figure*}}
	In order to convert \eqref{sinr} to a linear constraint,  \textcolor{black}{we adopt SCA with first order Taylor series expansion. 
	To this end, we define a
	new proxy function $g(\Phi_{f_{i},n_i,k})$ by dividing the product $\textbf{h}^{\dagger}_{f_{i},n_i,k}\textbf{w}_{f_{i},n_i,k}$ \textcolor{black}{into the}  real and imaginary parts as follows \cite{al2019energy1}:
 \begin{align}\nonumber
 g(\Phi_{f_{i},n_i,k})\triangleq&\mid\textbf{h}^{\dagger}_{f_{i},n_i,k}\textbf{w}_{f_{i},n_i,k}\mid^2=\\&\|[\Re(\textbf{h}^{\dagger}_{f_{i},n_i,k}\textbf{w}_{f_{i},n_i,k})\Im(\textbf{h}^{\dagger}_{f_{i},n_i,k}\textbf{w}_{f_{i},n_i,k})]^{T}\|^2,\label{procyfunction}
 \end{align}
 where $\Phi_{f_{i},n_i,k}=[\Re(\textbf{h}^{\dagger}_{f_{i},n_i,k}\textbf{w}_{f_{i},n_i,k})\Im(\textbf{h}^{\dagger}_{f_{i},n_i,k}\textbf{w}_{f_{i},n_i,k})]$. Note that \eqref{procyfunction} can easily be driven by the corresponding definition.
 For the sake of notation simplicity, we also define 
  \begin{align}\label{proxydef}
 g(\Phi_{f,n,k}^{f',n',k'})&\triangleq\mid\textbf{h}^{\dagger}_{f,n,k}\textbf{w}_{f',n',k'}\mid^2.
 \end{align}
 Now by substituting \eqref{proxydef} into \eqref{sinr}, we can rewrite \eqref{sinr} by \eqref{convert_prox}.
 \begin{figure*}
\begin{align}\label{convert_prox}
\frac{g(\Phi_{f_{i},n_i,k}^{(t)})}{\big(\lambda_{f_{i},n_i,k}-1\big)}\ge \sum_{k'\in\mathcal{K}_{\text{Total}}\backslash k,k'\succ k}
g(\Phi_{f_{i},n_i,k}^{f_{i},n_i,k'})+ \sum_{f'_{i}\in\mathcal{F}_{i}\backslash f_i}\sum_{k'\in\mathcal{K}_{\text{Total}}\backslash k'}	g(\Phi_{f_{i},n_i,k}^{f'_{i},n_i,k'})+\sigma^{2}_{f_i,n_i,k},
\end{align}\hrule
\end{figure*}
Both left and right sides of \eqref{convert_prox} are non-convex. To tackle this problem, we apply the first order Taylor series expansion to $g(\Phi_{f_{i},n_i,k})$. With considering two more considerable terms, we can approximate $g(\Phi_{f_{i},n_i,k})$ at  iteration $t$ ($t$ is the iteration number in SCA) by \cite{al2019energy1}
  \begin{align}\label{approxteylor}
 g(\Phi_{f_{i},n_i,k}^{(t)})\approxeq g(\Phi_{f_{i},n_i,k}^{(t-1)})+2(\Phi_{f_{i},n_i,k}^{(t-1)})^{T}\big[\Phi_{f_{i},n_i,k}^{(t)}-\Phi_{f_{i},n_i,k}^{(t-1)}\big].
 \end{align}
  Generalized version  of \eqref{approxteylor} according to the definition of \eqref{proxydef} is given by \eqref{genral_approxteylor}.
  \begin{figure*}
    \begin{align}\label{genral_approxteylor}
  g(\Phi_{f,n,k}^{f',n',k',(t)})\approxeq g(\Phi_{f,n,k}^{f',n',k',(t-1)})+2\big(\Phi_{f,n,k}^{f',n',k',(t-1)}\big)^{T}\big[\Phi_{f,n,k}^{f',n',k',(t)}-\Phi_{f,n,k}^{f',n',k',(t-1)}\big].
  \end{align}
  \hrule
\end{figure*}
   Moreover, we should approximate constraint \eqref{convert_prox} at point $\lambda_{f_{i},n_i,k}$ at iteration $t$ as follows:
 \begin{align}\nonumber
 \label{landaapprox}
\frac{ g(\Phi_{f_{i},n_i,k}^{(t)})}{\big(\lambda_{f_{i},n_i,k}^{(t)}-1\big)}&\approxeq\frac{ g(\Phi_{f_{i},n_i,k}^{(t-1)})}{\big(\lambda_{f_{i},n_i,k}^{(t-1)}-1\big)}\\&-\frac{ g(\Phi_{f_{i},n_i,k}^{(t-1)})}{{(\lambda_{f_{i},n_i,k}^{(t-1)})}^{2}}[\lambda_{f_{i},n_i,k}^{(t)}-\lambda_{f_{i},n_i,k}^{(t-1)}].
 \end{align}
By incorporating and applying these approximations into \eqref{convert_prox}, the convexity of constraint \eqref{opt24} is achievable. Hence, constraint \eqref{opt24}, i.e., \eqref{sinr} (or \eqref{convert_prox}) is replaced with the proposed linear approximation which is given by \eqref{Final_SCA} in which $g(\Phi_{f_{i},n_i,k}^{(t)})$ and   $\boldsymbol{\lambda}^{(t)}=[\lambda_{f_{i},n_i,k}]$ are updated according to \eqref{landaapprox} and \eqref{approxteylor}, respectively.
 \begin{figure*}[t]
	\begin{align}\label{Final_SCA}
\frac{ g(\Phi_{f_{i},n_i,k}^{(t-1)})}{\big(\lambda_{f_{i},n_i,k}^{(t-1)}-1\big)}-	\frac{ g(\Phi_{f_{i},n_i,k}^{(t-1)})}{{(\lambda_{f_{i},n_i,k}^{(t-1)})}^{2}}[\lambda_{f_{i},n_i,k}^{(t)}-\lambda_{f_{i},n_i,k}^{(t-1)}]\ge
		\sum_{k'\in\mathcal{K}_{\text{Total}}\backslash k,k'\succ k}
	\underbrace{g(\Phi_{f_{i},n_i,k}^{f_{i},n_i,k',(t)})}_{\text{Updated by \eqref{genral_approxteylor}}}+ \sum_{f'_{i}\in\mathcal{F}_{i}\backslash f_i}\sum_{k'\in\mathcal{K}_{\text{Total}}\backslash k'}	\underbrace{g(\Phi_{f_{i},n_i,k}^{f'_{i},n_i,k',(t)})}_{\text{Updated by \eqref{genral_approxteylor}}},	
	\end{align}
	\hrule
\end{figure*}}

	\textcolor{black}{For constraints \eqref{bsass}-\eqref{subass},
	we apply the change of variable as defining $\bold{w}_{f_i,n_i,k}\triangleq\exp(\hat{\bold{w}}_{f_i,n_i,k})$. \textcolor{black}{}Then, by replacing the aforementioned changes and approximations, we have the following optimization problem:
	\begin{subequations}\label{opt30}
		\begin{align}\label{opt31}
	&	\max_{\hat{\bold{W}},\boldsymbol{\theta},\boldsymbol{\lambda}}O(\hat{\bold{W}},\boldsymbol{\theta},\boldsymbol{\lambda}),\,\,\,\,\,
		\\\nonumber&\bold{s.t.}\, 
		\\\label{opt32}&
		G^\text{Total}_v
		\geq C^\text{Total}_v
		,\forall v\in\mathcal{V},
		\\\label{opt36}
		&
		 \sum_{k\in\mathcal{K}_{\text{Total}}}\sum_{n\in N}e^{\|\hat{\bold {w}}_{f_i,n_i,k}\|^2}\le e^{P_{\max}^{\text{FBS}}},~
		 \forall i\in\mathcal{I},\, f_i\in\mathcal{F}_i,\ f_i\neq 0,
		\\\label{opt37}& 
		\sum_{k\in\mathcal{K}_{\text{Total}}}\sum_{n_i\in\mathcal{N}_i}e^{\|\hat{\bold {w}}_{f_i,n_i,k}\|^2}\le e^{P_{\max}^{\text{MBS}}},\,\forall i\in \mathcal{I},\,f_i=0,
		\\&\label{opt310}\nonumber
		e^{\hat{\bold {w}}_{f_i,n_i,k}+\hat{\bold {w}}_{f'_i,n'_i,k}}\preceq e^{\bold{0}}=\bold{1},
		\\& \forall f_i,f'_i\in\mathcal{F}_i,\,n_i\in\mathcal{N}_i,n'_i\in\mathcal{N}_i,k\in\mathcal{K}_{\text{Total}}, f_i\neq f'_i,
		\\&\label{opt312}\nonumber
		e^{\big(\hat{\bold {w}}_{f_i,n_i,k}+\hat{\bold {w}}_{f_j,n_j,k}\big)}\preceq e^{\bold{0}}=\bold{1},
		\\& \forall f_i\in\mathcal{F}_i\backslash f'_i,\,n_i\in\mathcal{N}_i,n'_i\in\mathcal{N}_i,k\in\mathcal{K}_{\text{Total}},
		\\\label{opt311}\nonumber&
	e^{	\big(\hat{\bold {w}}_{f_i,n_i,k}+\hat{\bold {w}}_{f_i,n_i,c}+\hat{\bold {w}}_{f_i,n_i,q}\big)}\preceq e^{\bold{0}}=\bold{1}, 
	\\&\forall f_i\in\mathcal{F}_i,\,n_i\in\mathcal{N}_i,k, c, q\in\mathcal{K}_\text{Total}, k\neq c\neq q,
	\\&\nonumber
	\eqref{opttild},\eqref{opt23}, \eqref{opt25}, \eqref{opt28},   \eqref{Final_SCA},	
		\end{align}
	\end{subequations}
	where $\bold{\hat{\bold {W}}} = [\hat{\bold{w}}_{f_i,n_i,k}]$, $\boldsymbol{\lambda}=[\lambda_{f_{i},n_i,k}]$, $\boldsymbol{\theta}=[\theta_{f_i,n_i,k}^{\text{tilt}}].$
The optimization problem \eqref{opt30} is convex one and can be solved by utilizing the MATLAB optimization Toolbox such as CVX,  efficiently. The main steps of solution \eqref{opt30} is stated in Al. \ref{SCA}.} 
		\begin{algorithm}	\DontPrintSemicolon
		\renewcommand{\arraystretch}{0.9}
		\caption{SCA-based solution of problem \eqref{opt30}}
		\label{SCA}
		\KwInput{ Initialize $\boldsymbol{\Phi}^{0}$ and $\boldsymbol{\lambda}^{0}$, set
			$t=0$, 
						$0<\Upsilon<<1$ is tolerance (or accuracy), and $T$ is the allowable number of iterations
		}
		\Repeat{$\vert\vert \hat{\bold{W}}^{t}-\hat{\bold{W}}^{t-1}\vert\vert<\Upsilon$ or $t=T$}
		{
			Compute
						\eqref{approxteylor} and \eqref{landaapprox},	
			
			 Update $\boldsymbol{\lambda}^{t}$, $\boldsymbol{\Phi}^{t}$ (Via CVX)
			
			$t=t+1$
		}
		\KwOutput{ $\hat{\bold{W}},
			 \boldsymbol{\theta}, \boldsymbol{\lambda}$
		}
	\end{algorithm}

\section{{Convergency and Complexity Analysis}}
In this section, the convergence and complexity of iterative solution are stated.
\subsection{{Convergence of SCA} }
The SCA method (Al. \ref{SCA}) produces a sequence of feasible solutions by solving the convex problem in each iteration and the approximate terms are updated in the next iteration and then converges \cite{al2019energy, mokari2016limited}. To ensure convergence of the algorithm, two main conditions should be satisfied. First, initial setting of the network should be in the feasible set, i.e., satisfy the optimization problem constraints. Second, the objective function should be improved in each iteration until the predefined convergency condition in Al. \ref{SCA} is held. 
The convergence of solution against the number of iterations that are required to meet the predefined convergence conditions is illustrated in Fig.\ref{scacon}. From this figure, we can see that after
some iterations, the objective function is fixed. 
	  	\begin{figure}
	\centering
	\includegraphics[width=.48\textwidth]{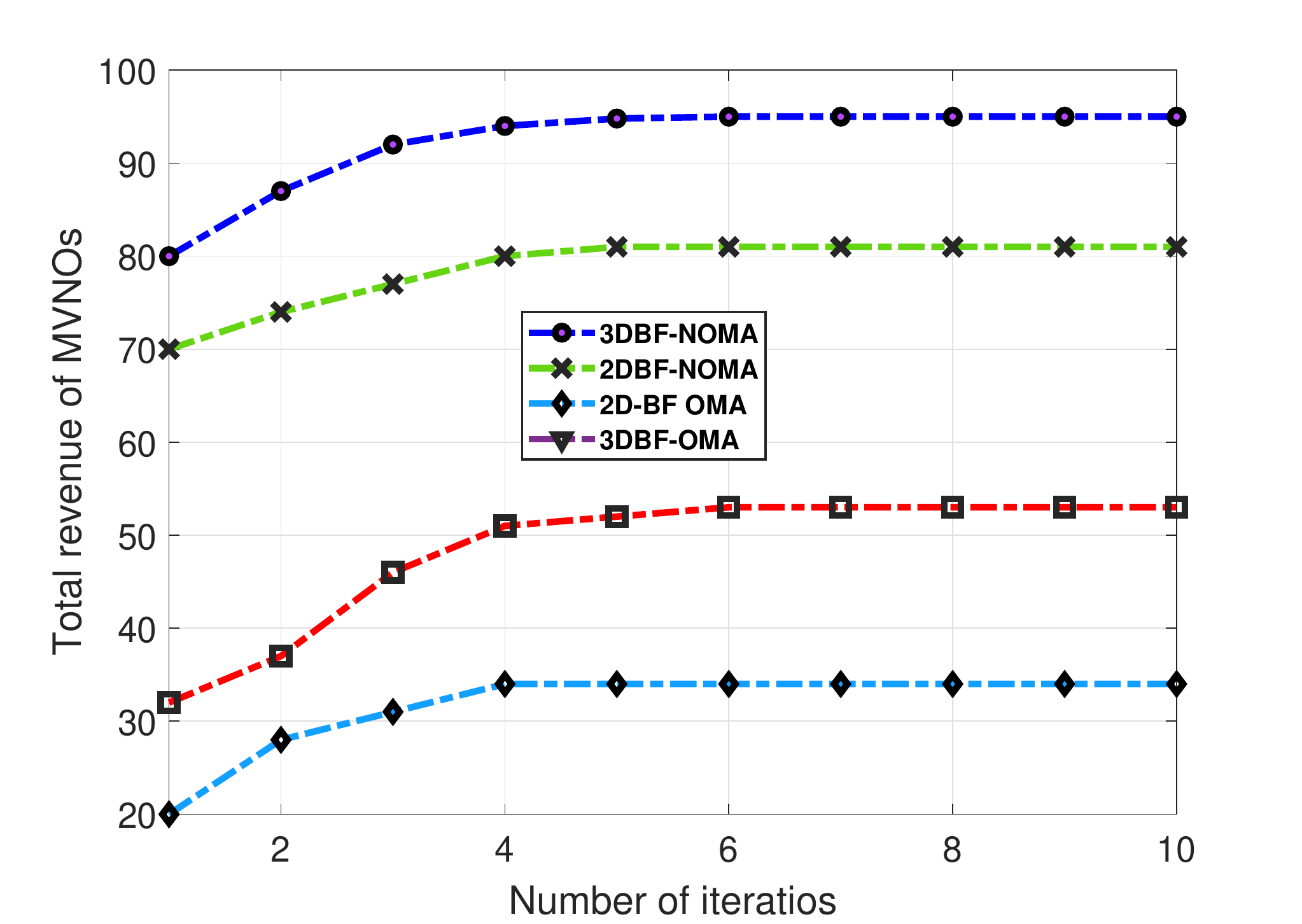}
	\caption{An examples of convergence of SCA algorithm for different network configuration.}
	\label{scacon} 
\end{figure} 
\subsection{{Computational Complexity}}
In this subsection, we analyze   the computational complexity of the proposed solution. 
 The complexity of an iterative algorithm is defined as the total number of iterations that is required  to converge the algorithm. Based on our proposed solution, the main computational complexity comes from solving the problem via CVX solver by applying geometric programming and interior point method (IPM) \cite{alavi2017robust}, \cite{mokari2016limited}. Based on this method, the overall computation complexity is given by 
 \begin{align}
 \varpi =\frac{\log\left(\frac{\Omega}{t^0\zeta}\right)}{\log(\vartheta)},
 \end{align} 
  where 
  \begin{align}\label{nom_cons}
    \Omega=7\times F\times K\times N+2\times V+ F,
  \end{align}
 is the total number of constraints of problem \eqref{opt30}, $t^0$ is the initial point for approximating the accuracy of IPM, $0 <\zeta<< 1$ is the stopping criterion for IPM and $\vartheta$ is used for updating the accuracy
of IPM \cite{boyd2004convex}, \cite{grant2008cvx}, \cite{8686217}. 
	\section{{Numerical Results}}\label{simulations}
	\subsection{{Simulation Environment}}
	In this section, the system performance of the proposed solution is evaluated under different aspects such as the
	 number of users and FBSs with Monte-Carlo simulations with $300$ iterations. In the numerical results, the system parameters are set as follows:\\
	\textcolor{black}{We assume two different InPs, each of which consists
	of a single MBS and $4$ FBSs randomly located in the main area with radius $0.5$ km and $20$ m, respectively \cite{6678362}. Moreover,
	the bandwidth of each InP is $5$ MHz \cite{7397887}.
 The frequency bandwidth of each subcarrier 
  is assumed to be $200$ KHz. Hence, the total number of subcarriers for each InP is $25$. 
 It is also assumed that there are $50$ single-antenna users are randomly distributed in the coverage area of the network at different distances. Moreover, $\Gamma_{f_i,n_i,k}^m$ is generated according to the complex Gaussian
 	distribution with mean $0$ and variance $1$, i.e., $ \boldsymbol{\Gamma}_{f_i,n_i,k}\thicksim\mathcal{CN}(	\bold {0}_{M_T},\bold{I})$, where $\bold{I}$ is the identity matrix.}
	\begin{table}
		\centering
		\caption{Configuration of network and setup parameters.}
		\label{Sim_Set}
		\begin{tabular}{|c|c|}
			\hline
			\textbf{Parameters(s)}  & \textbf{Value(s)}  \\
			\hline
			$I$ & $2$
			\\
			\hline
			 $ V$ & $5$
			 			\\
			 \hline
			 $F_i$ &    $5$
			\\
			\hline
			$\text{MBS radious}$, $\text{FBS radious}$ &    $500$ m, $20$ m
						\\
						 \hline
			$\text{MBS height}$, $\text{FBS height}$ &    $32~$ m, $2~$ m
			\\
						\hline
			$P_{\max}^{\text{MBS}}$, 	$P_{\max}^{\text{FBS}}$ &    $40$$\,\text{Wattss}$, $2 \,\text{Watts}$
			\\
			\hline
			$\text{BW}_i$, $N_i$ &    $5$~MHz, $25$
			\\
			\hline
			$K_{\text{Total}}$, $K_{v}$ &    $50$, $10$
			\\
			\hline
			$M_T$ &    $\text{min=5}, \text{max}=13$
			\\
			\hline
			$L$   &
			$2$
						\\
			\hline
			$\hat{\phi}_{3dB}$, $\hat{\theta}_{3dB}$, $\hat{\phi}_{0}$   &	$65\degree$, $65\degree$, $90\degree$
			\\ 
						\hline
			$	\theta_{f_i,n_i,k}^{\text{tilt},\min}$, $	\theta_{f_i,n_i,k}^{\text{tilt},\max}$   &	$0$, $\pi/2$
			\\
			\hline
		$G_{\max,\text{E}}$, $\text{SLL}_{\text{Al}}$, $\text{SLL}_{El}$    &	$8~$ dBi, $30$ dB, $30$ dB
		\\
		\hline
								$\text{Power spectral density of AWGN noise}$    &	$-174$ $\frac{\text{dBm}}{\text{Hz}}$
			\\
			\hline
											$\text{Channel coeffient}$    &	$ \text{Rayleigh distribution}$ \\&$ \text{with parameter}$ $1$
			\\
			\hline
				$\varrho$    &	$3$
			\\
			\hline
		\end{tabular}\label{simval}
	\end{table}%

	The power spectral density of the received AWGN noise is also set to $-174$ dBm/Hz.
	Moreover, $G_{\max,\text{E}}$,
		$\text{SLL}_{\text{Az}}$, 
	and $\text{SLL}_{\text{El}}$ 
	are set to $8$ 
	dBi,
	 $30$ dB, and $30$ dB, respectively \cite{nadeem2018elevation},\cite{kammoun2018design}, \cite{kammoun2014preliminary}. 
	In addition, we set $R_{\min}^{v}=30$ bps/Hz. Furthermore, we assume that $C_{i,f_i}=200$, and $G_{v,k}=10^{-6}$.
	  The value of other parameters are summarized in Table \ref{simval}. 
	  \subsection{{Simulation Results}}
	   The simulation analyses are provided from two main aspects; 1) The variants of different network parameters, 2) the different solution algorithms; as follows:
	  \subsubsection{{Various System Parameters}}
	  $\bullet~~${\textit{The Number of Users}}
	
	  Fig. \ref{userno} illustrates the total MVNO's revenue versus the total
	  number of users in the network. 
	  It is clear that, 3DBF with NOMA significantly improves the MVNO's revenue due to increasing achievable data rate of users and   system throughput. Moreover, by increasing the number of users, the MVNO's revenue is increased. That means the  MVNO with more customers obtains more revenues in compared to the others with low the number of customers.
	   Hence, not only advanced technologies like multiple antenna BF and NOMA have major effects on the enhancement of system spectral efficiency (throughput), but also it has considerable impact on the revenue of  operators and the costs of each InP.
	  \\$\bullet~~${\textit{The Number of FBSs}}
	  
	  Fig. \ref{numberofbs} shows the total MVNO's revenue versus the number of FBSs in each InP network by considering the number of users is
	  $35$. Clearly, by increasing the number of BSs, MVNO's revenue is improved. This is due to each MVNO's user has better experience and improves the achievable data rate. Hence, users pay more money to their MVNOs. Moreover, in this case probability of reducing path loss is increased and power consumption and its cost are reduced. 
	  On the other hand, in this scenario, the system throughout are increased. \textcolor{black}{ Hence, the total revenue is directly proportional to the system data rate, also increased.} 
	  	\begin{figure}
	  		\centering
	  		\includegraphics[width=.53\textwidth]{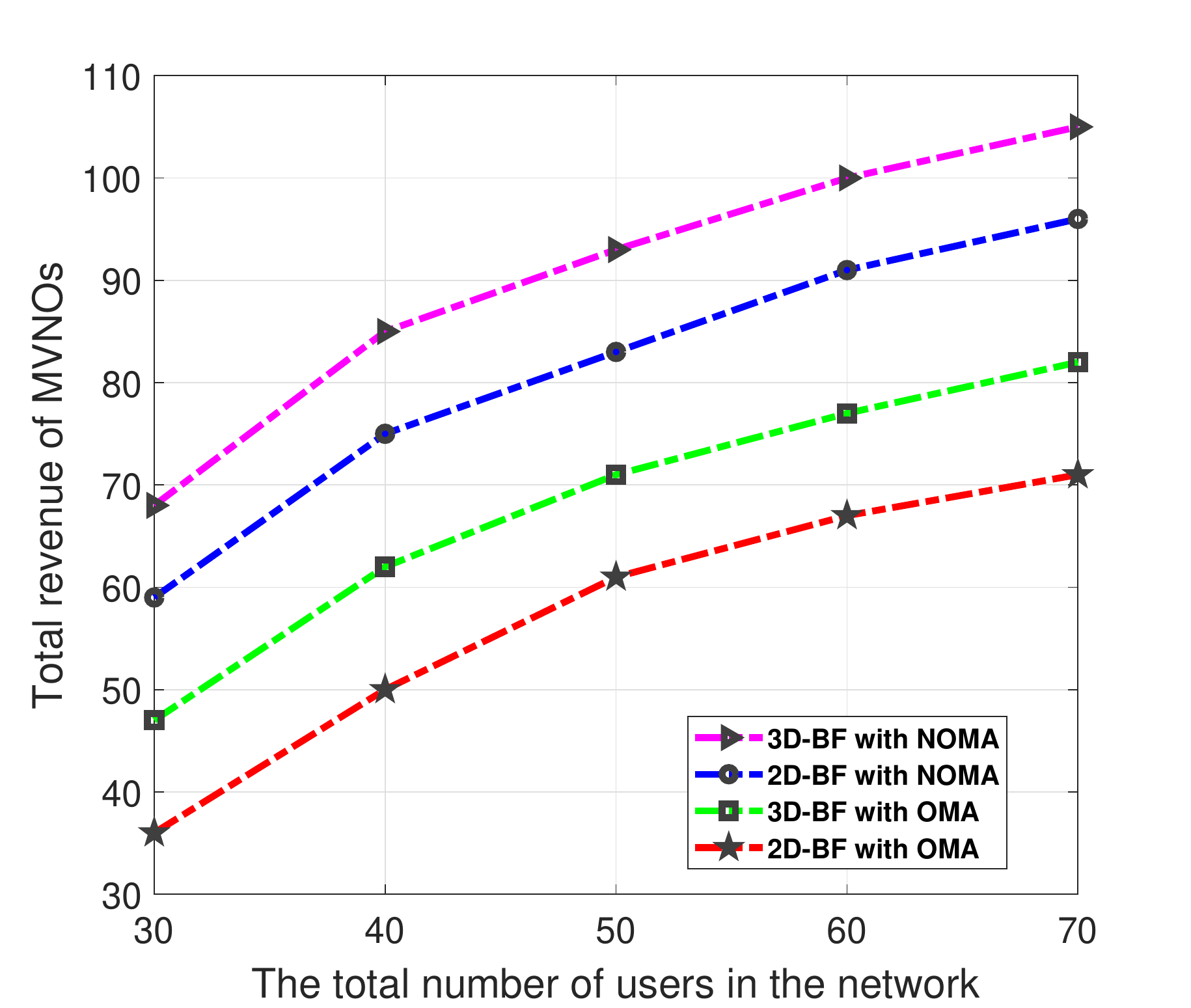}
	  		\caption{Total Revenue of MVNOs versus the total number of users in networks, with $P_{\max}^{\text{MBS}}=40$~Watts,  $P_{\max}^{\text{FBS}}=5$~Watts, $F_i=5$, and $M_T=5$.}
	  		\label{userno} 
	  	\end{figure}  
  	\\$\bullet~~${\textit{The Number of Transmit Antennas in Each BS}}
  		
 The effect of number of transmit antenna  ($M_{T}$) in each BS is illustrated in   Fig. \ref{txantenna}. As seen, by increasing the  number of transmit antenna in each BS, the total revenue of MVNOs is improved. This is due to the effect of antennas on the system throughput and energy efficiency. \textcolor{black}{ More important, in this scenario, 3DBF significantly outperforms 2DBF.} As a result, massive antennas with 3DBF system is a key enabler for massive connection networks in case of spectrum is limit or accusation of it has multiple times cost. Moreover,  3DBF reduces inter-cell interference and enhances the system performance from the SE perspective.
  \begin{figure}
  		\centering
  		\includegraphics[width=.53\textwidth]{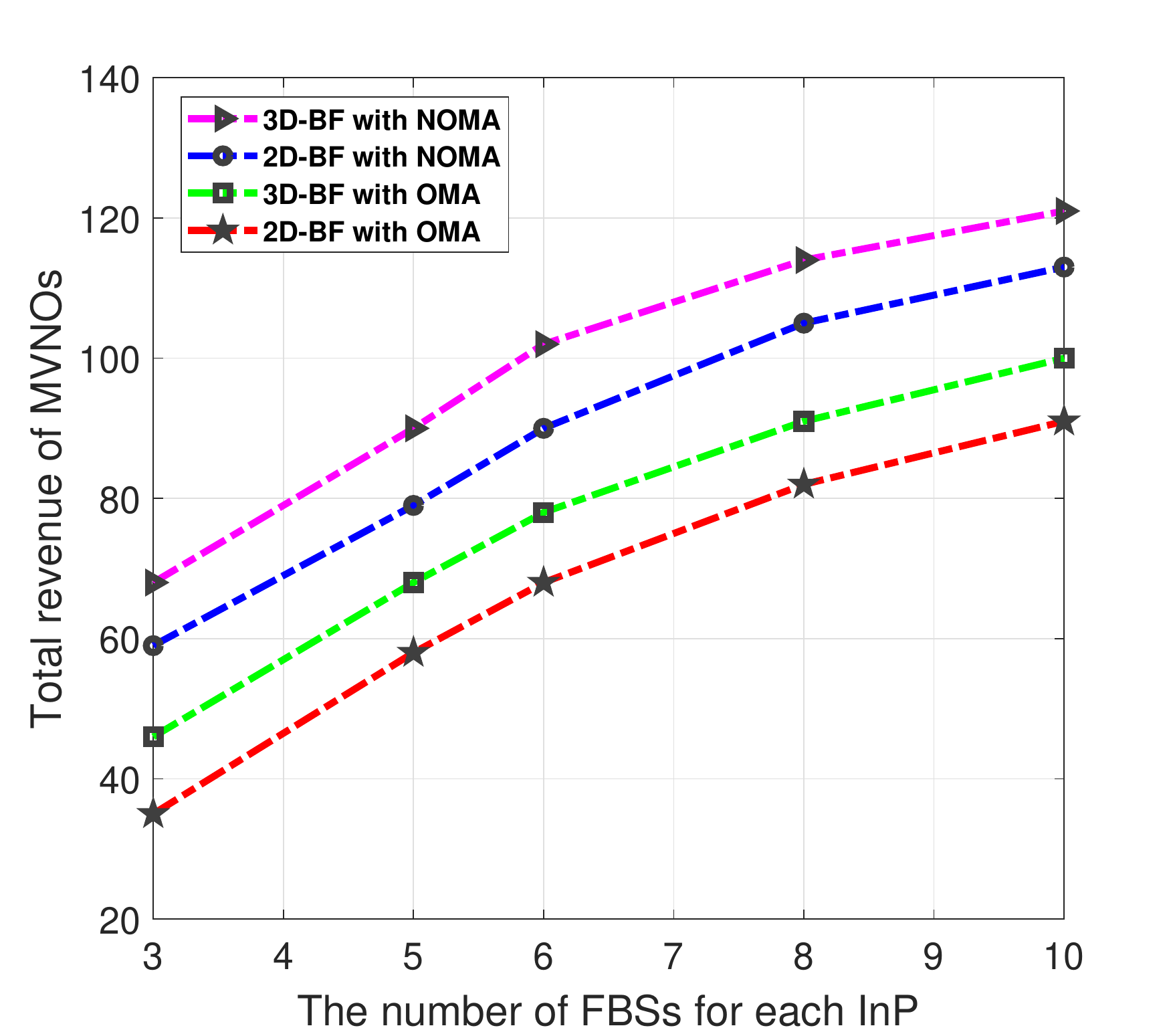}
  		\caption{Total Revenue of MVNOs versus the  total number of FBs in each InP, with $P_{\max}^{\text{MBS}}=40$~Watts,  $P_{\max}^{\text{FBS}}=5$~Watts, $K=50$, and $M_T=5$. }
  		\label{numberofbs} 
  	\end{figure}
\begin{figure}
	\centering
	\includegraphics[width=.52\textwidth]{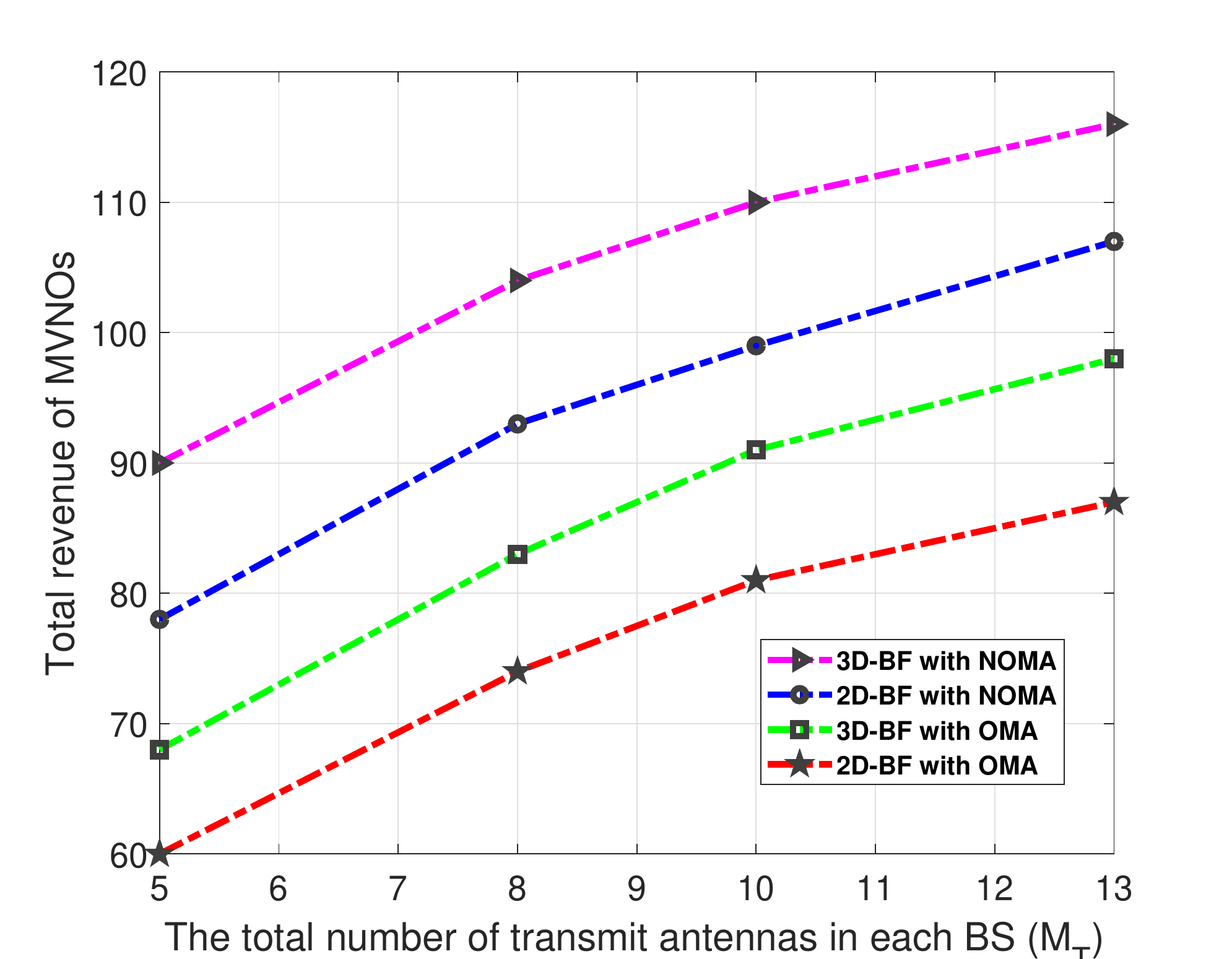}
	\caption{The achieved revenue against the total number of transmit antennas of each BS,  with $P_{\max}^{\text{MBS}}=40$~Watts,  $P_{\max}^{\text{FBS}}=5$~Watts, $K=50$, and $F=5$. }
	\label{txantenna} 
\end{figure}
\\$\bullet~~${\textit{The Total Bandwidth of Each InP }}

In Fig. \ref{subcarrier}, we investigate the total revenue of MVNOs versus maximum available bandwidth of each BS. Note that in this figure, subcarrier spacing is equal to $200$~\text{KHz}. As a result from this analyze, more spectrum for each MVNO makes that it has more revenue. Moreover, in this case NOMA has considerable performance improvement than that of 3DBF. In other words, in case of availability of spectrum, utilizing NOMA has more effects on the revenue compared to enabling 3DBF.   
\begin{figure}
	\centering
	\includegraphics[width=.53\textwidth]{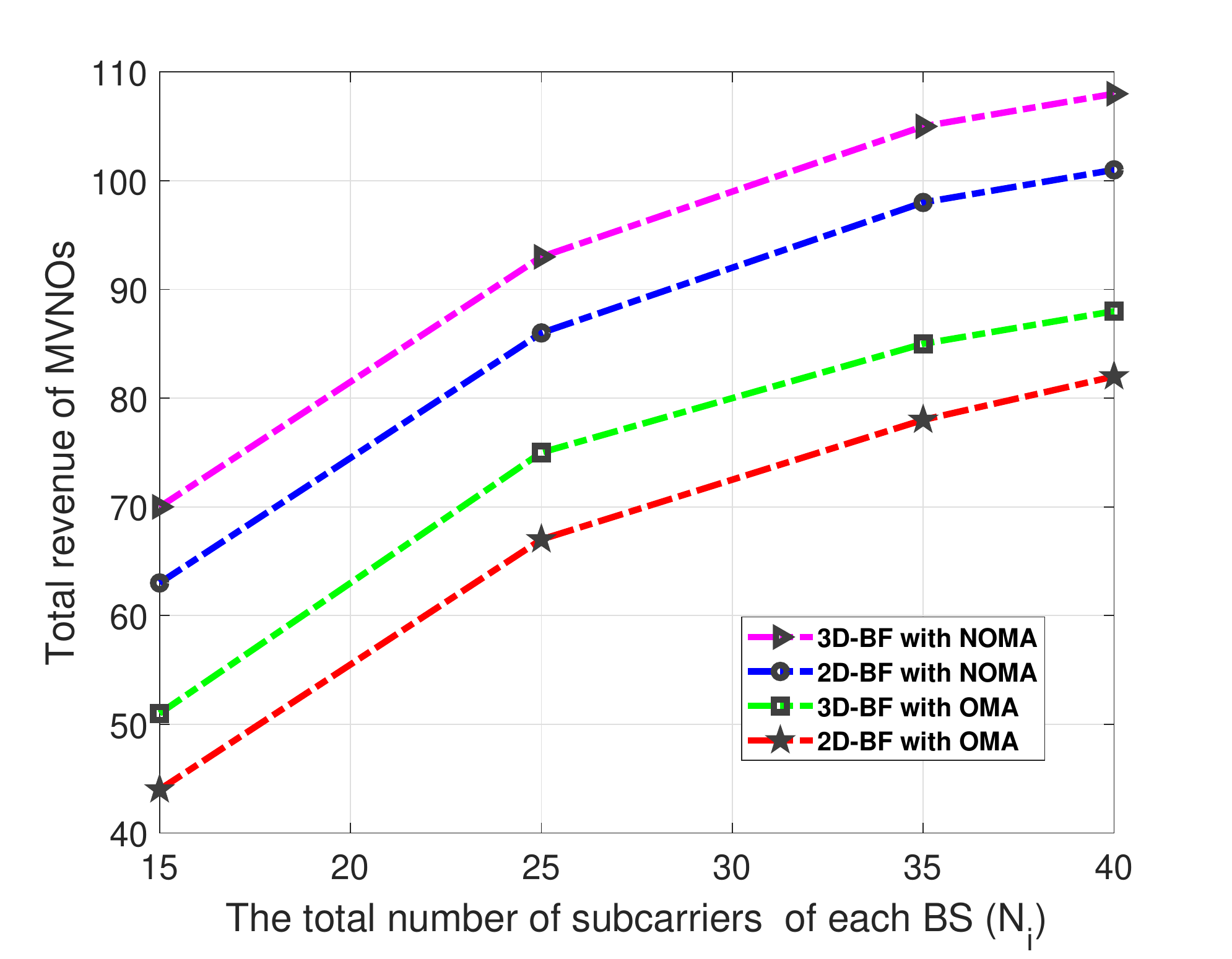}
	\caption{The achieved revenue against the total maximum available bandwidth  of each BS,  with $P_{\max}^{\text{MBS}}=40$~Watts,  $P_{\max}^{\text{FBS}}=5$~Watts, $K=50$, $F=5$, and $M_T=5$. }
	\label{subcarrier} 
\end{figure}
\\$\bullet~~${\textit{The Power Budget and Cost Weight}}

Furthermore, we investigate the maximum allowable power of MBSs impact  on the total revenue of MVNOs with different cost weights, i.e., $C_{i,f_i}$ in Fig. \ref{Power_Eff}. This figure depicts the comparison between different cost weights for power consumption in the InPs network. As a result from this figure, with NOMA and 3DBF the total revenue of MVNOs \textcolor{black}{improve up with approximately $54$\%}  compared to 2DBF with OMA. This is because, not only utilizing NOMA allows that each resource block, i.e., subcarrier can be exploited more \textcolor{black}{than one times without imposing} inter-cell interference, but also 3DBF massive antenna systems has a major improvement on the sum rate of users without more power consumption. \textcolor{black}{ Hence, based on the high income of MVNOs and the cost reduction of power consumption, the total revenue of MVNOs  increases.  
 Moreover, Fig. \ref{Power_Eff} highlights that high cost weights for power consumption decreases the total revenue of MVNOs in contrast to power budget enhancement. }
 
 We emphasize the performance and complexity comparison of the considered techniques  in Table \ref{Com_Order_Tec}, \textcolor{black}{based on the parameters shown in Fig. \ref{userno} as an example of configuration.}  As seen, exploiting 3DBF and NOMA improves the total revenue by approximately    $47$\% compared to 2DBF with OMA. This is because of, in 3DBF NOMA network, the SINR of each user is improved by adjusting the the beam pattern of antennas as desired and with NOMA the system throughput is increased. On the other hand, \textcolor{black}{our proposed revenue is} based on the received data rate of users. Hence, these results are observed. 
 \\$\bullet~~${\textit{Deployment Scenarios }}
 
 Based on the 3-rd  generation partnership project (3GPP) standardization, from geographical locations perspective, we have $4$ main areas as deployment scenarios that are presented in Table \ref{Dep_Sce} \cite{22261}, \cite{NGMN}, \cite{ITU}. Based on this, \textcolor{black}{we to evaluate the proposed problem in these scenarios} to verify its practicality. In this regard, we configure the network according to the characteristics of these scenarios, \textcolor{black}{ in which the considered system parameters are  depicted in Table \ref{BS_con} for a small scaled network. Table \ref{Per_Dep_com} compares the performance of 2DBF and 3DBF for different access technologies and various deployment scenarios in Table \ref{Dep_Sce}. In this Table, we consider 2DBF with OMA access is baseline and comparison values are listed with defined ratio as $\frac{\text{Mentioned technology name}}{\text{2DBF OMA}}$, as an example $\frac{\text{3DBF NOMA}}{\text{2DBF OMA}}$. From this results, we obtain that 3DBF with NOMA is the most beneficial option for Indoor hotspot locations.} This is has two  main reasons, the first is the efficiency of NOMA, since the number of users is high and SE of NOMA is more than \textcolor{black}{that of the other scenarios. The second is refereed to the performance of 3DBF in which it improves the received SINR of users by assigning pencil beams with high gains to them.} Moreover, in these scenarios, the data rate of users that are located in different floors and blind locations are significantly improved.

 \begin{table*}[t]
 	\centering
 	\caption{Performance comparison of different  technologies in different geographical locations (deployment scenarios) based on the setting in Table \ref{BS_con}.}
 	\begin{tabular}{p{.1cm} p{.15cm}|p{2.5cm}|p{2.5cm}|p{2.5cm}|}
 		\toprule
 		\midrule
 		\multicolumn{5}{c}{Technologies}\\ 
 		\cmidrule{3-5}
 		&& 2DBF with NOMA  &  3DBF with OMA& 3DBF with NOMA\\ 
 		\cmidrule{3-5}
 		\multicolumn{1}{c}{\multirow{4}{*}{\begin{sideways}Scenarios\end{sideways}}}   &
 		\multicolumn{1}{l}{Indoor hotspot}& 2.7x & 2.3x & 4.6x\\
 		\cmidrule{2-5}
 		\multicolumn{1}{c}{}    &
 		\multicolumn{1}{l}{Dense urban}& 2.3x & 1.9x & 3.8x\\
 			\cmidrule{2-5}
 		\multicolumn{1}{c}{}    &
 		\multicolumn{1}{l}{Urban}& 1.7x  & 1.6x & 3.2x\\
 		\cmidrule{2-5}
 		\multicolumn{1}{c}{}    &
 		\multicolumn{1}{l}{Suburban/Rural} & 1.3x & 1.4x & 2.1x\\
 		\midrule
 	\end{tabular}
 	\label{Per_Dep_com}
 	\vspace{-0.2cm}
 \end{table*} 

$\bullet~~${\textit{Signaling overhead of 2D and 3D channels}}

Radio resource management in NR is based on the  measurement of reference signals such as synchronization and channel information, i.e., CSI \cite{bertenyi20185g}. Some of them are used for time and frequency synchronization in the random access procedure\footnote{Random access enables each user to access a cell.} and  demodulation of signals at receiver sides (via demodulation reference signal).
     In this paper, signaling overhead refers to the amount of feedback required for the 2D and 3D channels estimation. In the 3D channel,  channel information is required \textcolor{black}{for} both horizontal and vertical directions, in contrast to the 2D channel model that is just in the horizontal direction. As a result, the predefined overhead signal is increased rapidly in 3D by the number of antennas and \textcolor{black}{the} BSs in the network. We compare the 2D and 3D channels with signaling overhead metric in Table \ref{Overhead_Com}, where, assuming the number of vertical antennas is $M_{T}'$. We infer that  the ratio of signaling overhead of 3DBF over 2DBF is closed together in high order of transmit antennas.     
 
  \begin{table*}
 	\centering
 	\caption{Assumptions and main characteristics for deriving different deployment scenarios for 5G based on the 3GPP specifications.}
 	\label{Dep_Sce}
 	\begin{tabular}{ |c|c|c|c|l|}
 		\hline
 				&Overall User density ($\text{per}{\text{Km}}^{2}$)&  Active user data rate (Mbps)&Activity factor \\	
 		\hline	
 		Indoor& $250000$& $200$& $30\%$
 		\\
 		\hline
 		Dense Urban &25000 &300&$10\%$	
 		 		\\
 		\hline
 		Urban &10000 &50&$20\%$
 		\\
 		\hline
 		Rural &100 &50&$20\%$
 		\\
 		\hline
 	\end{tabular}
 \end{table*}
 
 \begin{table*}[t]
 	\centering
 	\caption{ Network configuration in different deployment scenarios, other parameters are based on Table \ref{Sim_Set}.}
 	\begin{tabular}{p{.1cm} p{1cm}|p{1.3cm}|p{2cm}|p{1.5cm}|p{2.3cm}|p{1cm}|}
 		\toprule
 		\midrule
 		\multicolumn{6}{c}{Deployment scenario}\\ 
 		\cmidrule{3-6}
 		&&Indoor& Dense urban  & Urban & Suburban/ Rural\\ 
 		\cmidrule{1-6}
 		\multicolumn{1}{c}{\multirow{4}{*}{\begin{sideways}Network Configuration \end{sideways}}}   &
 		\multicolumn{1}{l}{BS Height}&15 m& 25 m & 35 m & 45 m\\
 		\cmidrule{2-6}
 		\multicolumn{1}{c}{}    &
 		\multicolumn{1}{l}{Maximum transmit power} & 10 Watts &  20  Watts& 40 Watts & 60 Watts\\
 		\cmidrule{2-6}
 		\multicolumn{1}{c}{}    &
 		\multicolumn{1}{l}{The number of carriers} & 100 & 64  &  40&20\\
 		\cmidrule{2-6}
 		\multicolumn{1}{c}{}    &
 		\multicolumn{1}{l}{Inter-site distance} &20 m &200 m & 500 m & 5000 m  \\
 		 		\cmidrule{2-6}
 		\multicolumn{1}{c}{}    &
 		\multicolumn{1}{l}{Number of antennas ($M_T$)} &15& 12 & 8 &  5 \\
 		\cmidrule{2-6}
 		\multicolumn{1}{c}{}    &
 		\multicolumn{1}{l}{User Density} &200/$\text{km}^2$& 150/$\text{km}^2$ & 70/$\text{km}^2$ &  30/$\text{km}^2$ \\
 		\midrule
 	\end{tabular}
 	\label{BS_con}
 	\vspace{-0.2cm}
 \end{table*}

 \begin{table*}
 	\centering
 	\caption{Complexity order and performance comparison of 2DBF with OMA and 3DBF with NOMA according to Fig. \ref{userno}.}
 	\label{Com_Order_Tec}
 	\begin{tabular}{ |c|c|c|c|l|}
 		\hline
Metric 		&2DBF OMA&3DBF NOMA&3D-NOMA over 2D-OMA \\	
 		\hline	
 		Complexity&$\mathcal{O}\left(2V+F+F\times N\times K\right)$&$\mathcal{O}\left(2V+F+4\times F\times N\times K+ F\times N\right)$& $0.75 \% \uparrow$
 		\\
 		\hline
 		Objective &71 &105&$47\% \uparrow$
 		\\
 		\hline	
 	\end{tabular}
 \end{table*}
		\begin{figure}
	\centering
	\includegraphics[width=.52\textwidth]{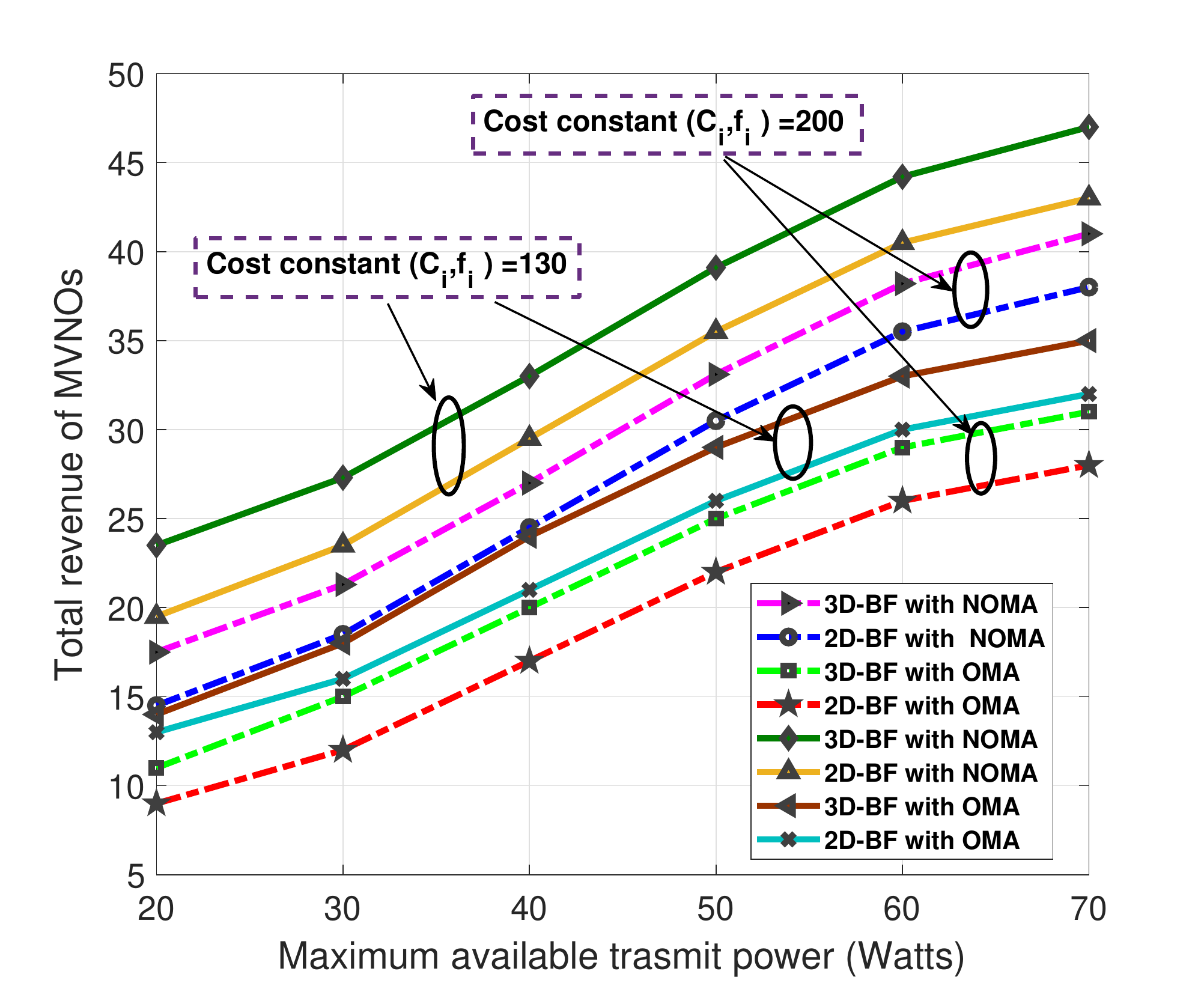}
	\caption{The total Revenue of MVNOs versus the maximum transmit power budget, with $U=30$~Watts, $N_i=25$, $F_i=5$ and $M_T=5$.}
	\label{Power_Eff} 
\end{figure}

\begin{table*}
	\centering
\caption{Signaling overhead comparison of the 2DBF and 3DBF.}
	\label{Overhead_Com}
	\begin{tabular}{ |c|c|c|c|}
		\hline
	Metric	&2DBF&3DBF&3DBF  over 2DBF \\	
		\hline	
		Overhead&$M_{T}\times F=13\times 10=130 $&$(M_{T}+M_{T}')\times F=(13+2)\times 10=150$& $ 15\% \uparrow$
		\\
		\hline	
	\end{tabular}
\end{table*}
\subsubsection{{Comparison Between the Proposed Solution Methods}}
In this subsection, we investigate different solutions, namely, JS-CIV, ASM, and optimal in which each of them is briefly discussed in the following.
\\$\bullet~~${\textit{JS-CIV}}

We propose JS-CIV to solve problem \eqref{opt30}. To this end, we convert non-convex main problem \eqref{opt0} to convex problem \eqref{opt30}. The main step is to convert the binary variable, i.e., $\boldsymbol{\rho}$ to the predefined optimization variable, i.e., $\bold{W}$. Then, by exploiting the change variable method and first Taylor approximation, we formulate problem \eqref{opt30}. Finally, we solve it with SCA (Al. \ref{SCA}).
Note that in this method, we solve all considered optimization variables jointly. 
 \\$\bullet~~${\textit{ASM}}
 
 Based on the ASM method, we divide main problem  \eqref{opt0} into two sub-problems, namely, joint subcarrier and user association and 3D beam allocation. Then, we solve each on them separately and iterate between them until stopping \textcolor{black}{criteria are met.} Note that in this method, we also use SCA-based difference of concave functions  and first order of Taylor approximation in the 3D beam allocation sub-problem.
 \\$\bullet~~${\textit{Optimal}}
 
  We devise the exhaustive search method to investigate  the optimality gap of the iterative sub-optimal solutions in the considered parameters. 
  Fig. \ref{solution} illustrates  the performance comparison of the aforementioned solutions. As seen, the optimality gap of the proposed solution, i.e., JS-CIV is nearly   $8$\%, while this gap is $15$\% for  ASM. This is because of disjoint solving of the integer and continues variables in ASM, in contrast to jointly solving of them with SCA.    Moreover, we compare the complexity order of the mentioned solutions that is stated in Table  \ref{Com_Order}. In this table, $\Omega$ is obtained from \eqref{nom_cons}, $E^{\text{Pos}}$, and $Q^{\text{Pos}}$ are different possible values in the exhaustive search method for $\bold{W}$ and $\boldsymbol{\theta}$, respectively. 
\begin{figure}
	\centering
	\includegraphics[width=.53\textwidth]{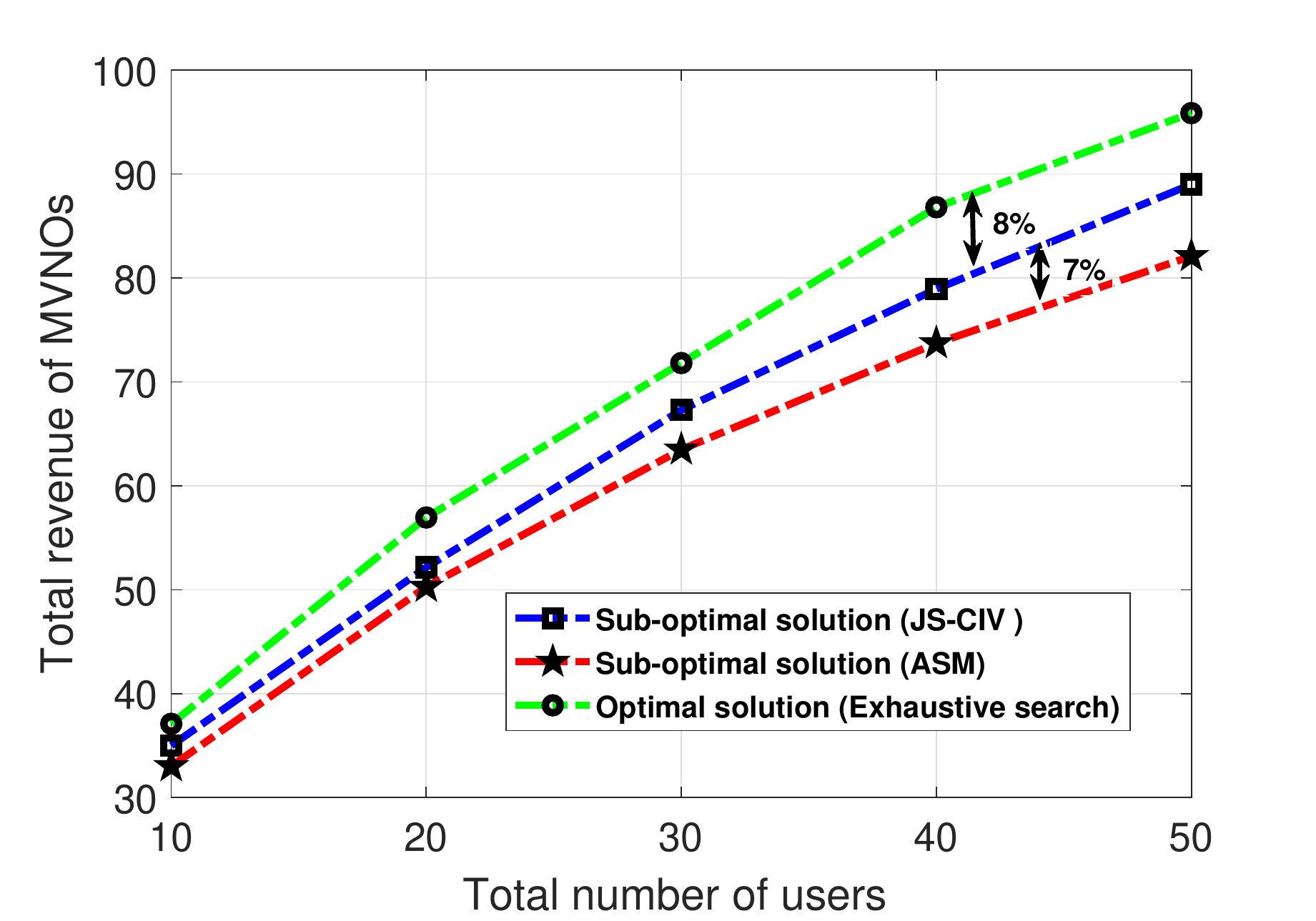}
	\caption{Total Revenue of MVNOs versus the total number of users in networks, with $P_{\max}^{\text{MBS}}=30$~Watts,  $P_{\max}^{\text{FBS}}=5$~Watts, $F_i=4$ and $M_T=5$ for different solutions methods.}
	\label{solution} 
\end{figure}
\begin{table*}
	\centering
	\caption{\textcolor{black}{Complexity order and performance comparison of the proposed solutions according to Table \ref{Sim_Set} and Fig. \ref{solution}.}}
	\label{Com_Order}
	\begin{tabular}{ |c|c|c|c|}
		\hline
		&JS-CIV&ASM	&Optimal \\	
		\hline	
		 Complexity& $ \varpi =\frac{\log\left(\frac{\Omega}{t^0\zeta}\right)}{\log(\vartheta)}$&  $ \frac{\log\left(\frac{2(V+ F\times N\times K)+N\times F\times I}{t^0\zeta}\right)}{\log(\vartheta)}+\frac{\log\left(\frac{2V+ F\times N(1+K)+F\times I}{t^0\zeta}\right)}{\log(\vartheta)}$& $(2E^{\text{Pos}})^{(U\times N)}\times (Q^{\text{Pos}})^{(U\times M\times F\times I)}$
		 \\
 		\hline
	Performance & $8\% \downarrow$& $15\%\downarrow$& Baseline
		\\
		\hline	
	\end{tabular}\label{Table signaling}
\end{table*}
	\section{{conclusion}}\label{conclusions} 
In this paper, \textcolor{black}{ we studied a novel  multi-user  MISO-3DBF-NOMA-based  HetNet from the cost and the system performance perspectives.} Our proposed optimization problem is based on the maximizing MVNO's total revenue under resource limitation and guaranteeing the user's required QoS and contains both integer and continues variables.  To solve the proposed optimization problem efficiently, we propose a new algorithm called JS-CIV where \textcolor{black}{  it solves  the main problem without adopting ASM, by exploiting the change variables method.} 
 Our evaluations, benchmarked against a baseline, demonstrated that the impact of  utilizing 3DBF in NOMA-based HetNet from the performance and MVNO's revenue perspectives. Furthermore, the convergence of the proposed algorithm is verified  in our evaluations. Our numerical results show that 3DBF has good performance in massive antennas systems. Also \textcolor{black}{NOMA is very appropriate for
 	ultra high density environments which require very high spectrum. Indeed, NOMA-3DBF significantly improves the system throughput and subsequently the MVNO's revenue can be improved with approximately $54$\% in massive connection networks. Moreover, we investigated the proposed system model \textcolor{black}{for} different deployment scenarios and signaling overhead as a metric. Moreover, JS-CIV outperforms ASM and its optimality gap is nearly $8$\%. As a future work, we aim to study
 	the proposed algorithm for imperfect CSI.}   

\hyphenation{op-tical net-works semi-conduc-tor}
\bibliographystyle{ieeetr}
\bibliography{citationMISOBeamforming}{}
\end{document}